\documentclass[vecphys]{svmult}
\usepackage{graphicx}        
\usepackage[bottom]{footmisc}

\def\a{\alpha}
\def\b{\beta}
\def\g{\gamma}
\def\d{\delta}

\def\th{\vartheta}

\def\l{\lambda}

\def\r{\rho}
\def\s{\sigma}
\def\t{\tau}
\def\ph{\varphi}
\def\o{\omega}

\def\D{\Delta}

\def\O{\Omega}
\def\half{\textstyle{1\over2}}

\def\quarter{\textstyle{1\over4}}
\def\dim{\rm dim}
\def\rank{{\rm rank}}
\def\and{\hbox{ and }}
\def\And{\quad\hbox{and}\quad}
\def\red{{\rm red}}
\def\black{{\rm black}}
\def\norm#1{\|#1\|}
\def\bar#1{{\overline{#1}}}

\def\sn{\smallskip\noindent}
\def\mn{\medskip\noindent}

\def\pn{\noindent}

\def\implies{\Longrightarrow}

\def\Implies{\quad\Longrightarrow\quad}

\def\A{{\cal A}}
\def\B{\ifmmode{\cal B}\else {$\cal B$} \fi}
\def\C{\ifmmode{\cal C}\else {$\cal C$} \fi}
\def\D{\ifmmode{\cal D}\else {$\cal D$} \fi}
\def\E{\ifmmode{\cal E}\else {$\cal E$} \fi}
\def\F{\ifmmode{\cal F}\else {$\cal F$} \fi}
\def\H{{\cal H}}
\def\I{\ifmmode{\cal I}\else {$\cal I$} \fi}
\def\J{\ifmmode{\cal J}\else {$\cal J$} \fi}
\def\K{\ifmmode{\cal K}\else {$\cal K$} \fi}
\def\L{\ifmmode{\cal L}\else {$\cal L$} \fi}
\def\M{\ifmmode{\cal M}\else {$\cal M$} \fi}
\def\N{\ifmmode{\cal N}\else {$\cal N$} \fi}
\def\P{\ifmmode{\cal P}\else {$\cal P$} \fi}
\def\R{\ifmmode{\cal R}\else {$\cal R$} \fi}
\def\V{\ifmmode{\cal V}\else {$\cal V$} \fi}
\def\X{\ifmmode{\cal X}\else {$\cal X$} \fi}
\def\Y{\ifmmode{\cal Y}\else {$\cal Y$} \fi}
\def\Z{\ifmmode{\cal Z}\else {$\cal Z$} \fi}
\let\SS=\S
\def\paragraph{\SS}
\def\S{\ifmmode{\cal S}\else {$\cal S$} \fi}
\def\T{\ifmmode{\cal T}\else {$\cal T$} \fi}
\def\U{\ifmmode{\cal U}\else {$\cal U$} \fi}
\def\Z{\ifmmode{\cal Z}\else {$\cal Z$} \fi}

\def\EE{\hbox{\rm I\hskip -0.5em E}}
\def\HH{\ifmmode{I\hskip -0.3em H}
    \else{\hbox{$I\hskip -0.3em H$}}\fi}
\def\NN{\hbox{\rm I\hskip -0.2em N}}
\def\Nn{\ifmmode{I\hskip -0.3em N_0}
    \else{\hbox{$I\hskip -0.3em N_0$}}\fi} 
\def\PP{\hbox{\rm I\hskip -0.5em P}}
\def\RR{\hbox{\rm I\hskip -0.5em R}}
\def\eins{{\mathchoice {\rm 1\mskip-4mu l} {\rm 1\mskip-4mu l}
                      {\rm 1\mskip-4.5mu l} {\rm 1\mskip-5mu l}}} 
\def\11{\ifmmode{\eins} \else {$\eins$} \fi}

\def\one{\eins}
\def\done{\ifmmode{\cdot\eins} \else {$\cdot\eins$} \fi}

\def\ZZm{{\mathchoice
         {\hbox{$\sans\textstyle Z\kern-0.4em Z$}}%
         {\hbox{$\sans\textstyle Z\kern-0.4em Z$}}%
         {\hbox{$\bf\scriptstyle Z$}}%
         {\hbox{$\bf\scriptscriptstyle Z$}}%
        }}                                         
\def\ZZ{{\ifmmode{\ZZm} \else {$\ZZm$} \fi}}

\def\CCm{{\mathchoice                                     
   {\setbox0=\hbox{$\displaystyle\rm C$}\hbox{\hbox to0pt{\kern0.4\wd0
                                           \vrule height0.9\ht0\hss}\box0}}
   {\setbox0=\hbox{$\textstyle   \rm C$}\hbox{\hbox to0pt{\kern0.4\wd0
                                           \vrule height0.9\ht0\hss}\box0}}
   {\setbox0=\hbox{$\scriptstyle \rm C$}\hbox{\hbox to0pt{\kern0.4\wd0
                                           \vrule height0.9\ht0\hss}\box0}}
   {\setbox0=\hbox{$\scriptscriptstyle\rm C$}\hbox{\hbox to0pt{\kern0.4\wd0
                                           \vrule height0.9\ht0\hss}\box0}}}}
\def\CC{{\ifmmode{\CCm} \else {$\CCm$} \fi}}

\def\TTm{{\mathchoice
    {\setbox0=\hbox{$\displaystyle\rm T$}
        \hbox{\hbox to0pt{\kern0.3\wd0\vrule height0.9\ht0\hss}\box0}}
    {\setbox0=\hbox{$\textstyle\rm T$}
        \hbox{\hbox to0pt{\kern0.3\wd0\vrule height0.9\ht0\hss}\box0}}
    {\setbox0=\hbox{$\scriptstyle\rm T$}
        \hbox{\hbox to0pt{\kern0.3\wd0\vrule height0.9\ht0\hss}\box0}}
    {\setbox0=\hbox{$\scriptscriptstyle\rm T$}
        \hbox{\hbox to0pt{\kern0.3\wd0\vrule height0.9\ht0\hss}\box0}}
        }}
\def\TT{\ifmmode{\TTm} \else {$\TTm$} \fi}

\def\sk#1#2{\sin^2(\a_{#1}-\b_{#2})}

\def\set#1#2{\bigl\{\>#1\>\big|\>#2\>\bigr\}}

\def\bra#1{\langle#1|}
\def\ket#1{|#1\rangle}
\def\inp#1#2{\langle#1,#2\rangle}
\def\Inp#1#2{\left\langle#1,#2\right\rangle}

\def\set#1#2{\bigl\{\>#1\>\big|\>#2\>\bigr\}}

\def\alg{{\rm alg}\,}
\def\norm#1{\left\|\,#1\,\right\|}
\def\norms#1{\left\|\,#1\,\right\|^2}

\def\ten{\otimes}

\def\son#1{\sum_{#1=1}^n}
\def\sok#1{\sum_{#1=1}^k}
\def\sol#1{\sum_{#1=1}^l}
\def\som#1{\sum_{#1=1}^m}
\def\sonij{\sum_{i,j=1}^n}

\def\tr{{\rm tr}\,}
\def\sp{{\rm sp}}
\def\id{{\rm id}\,}
\def\for{{\rm for}}
\def\where{\quad\hbox{where}\quad}
\def\with{\quad\hbox{with}\quad}
\def\tuple#1_#2{#1_1,#1_2,\ldots,#1_{#2}}
\def\tup#1_#2{#1_1,\ldots,#1_{#2}}
\def\At{{\widetilde{\cal A}}}
\def\Ht{{\widetilde{\cal H}}}

\def\pst{{\widetilde\psi}}
\def\Astat{{\A^*_{+,1}}}
\def\Bstat{{\B^*_{+,1}}}

\def\afbeelding#1#2#3#4{
\begin{figure}
\centering
\includegraphics[height=#1]{#2}
\caption{#3}
\label{#4}
\end{figure}}

\def\correspondence#1#2{\smallskip%
           \hbox{\vtop{\hsize=0.45\hsize{\noindent#1}\vfill}
                 \vtop{\hsize=0.05\hsize{\noindent \hfil \^= \hfil}\vfill}
                 \vtop{\hsize=0.5\hsize{\noindent#2}\vfill}
                }\smallskip}

\def\Heis#1from#2to#3{#3\mathop{\longleftarrow}\limits^{#1}#2}
\def\Schr#1from#2to#3{#2^*\mathop{\longrightarrow}\limits^{#1^*}#3^*}
\def\mapsby#1{{\mathop{\longmapsto}\limits^{#1}}}

\font\caps=cmcsc10

\begin{document}

\title*{Quantum Probability and Quantum Information Theory}
\titlerunning{Quantum Probability}
\author{Hans Maassen\inst{}}
\institute{Radboud University, Toernooiveld 1, 6525 ED Nijmegen, Netherlands,
\texttt{maassen@math.ru.nl}}
\maketitle

\section{Introduction}
\label{sec:Intro}
From its very birth in the 1920's,
quantum theory has been characterized by a certain strangeness:
its seems to run counter to the intuitions that we humans have about the
world we live in.

According to these `realistic' intuitions all things have their
definite place and sharply determined qualities, such as speed,
color, and weight. Quantum theory, however, refuses to precisely
pinpoint them. With respect to this apparent shortcoming of the
theory different points of view can be taken. It could be suspected
that quantum theory is incomplete, that it gives a coarse
description of a reality that is actually more refined. This is the
viewpoint once taken by Einstein, and it still has adherents today.
It calls for a search for finer mathematical models of physical
reality, based on classical probability,
often referred to as `hidden variable models'. One such
attempt is Bohm's theory of non-relativistic quantum mechanics.

\noindent
However, the work of John Bell in the 60's and of Alain
Aspect in the 70's and 80's strongly favors the opposite point of view: their
work has made clear that such models with a classical probabilistic
structure are necessarily afflicted with a certain weakness: they
must at least allow {\it action at a distance}. This we regard as a
bad property for a theory which aims to describe a physical world
where no signals have been observed to travel faster than light.
Apart from that, the hidden variable theories which have been found
so far are highly artificial and cannot be tested against quantum
mechanics since they do not predict any new phenomena.

\noindent It is for these reasons that we decide to accept
quantum theory with its inherent strangeness,
and are prepared to modify probability theory accordingly.


\subsection{Quantum Probability}
\label{sec:2} So quantum mechanics does not predict the results of
physical experiments with certainty, but calculates probabilities
for their possible outcomes.

Now, the classical mathematical theory of probability obtained a unified
formulation in the 1930's, when Kolmogorov introduced his axioms,
defining the universal structure $(\O,\Sigma,\PP)$ of a probability
space.
For a long time this theory of probability
(dealing with probability distributions, stochastic processes,
Markov chains, martingales, etc.) remained completely separate
from the mathematical development of quantum mechanics (involving
vectors in a Hilbert space, hermitian operators, unitary
transformations, and such like).

In the 1970's and 1980's people around Accardi, Lewis, Davies,
K\"ummerer, building on ideas of von Neumann's and Segal's,
developed a unified framework,
a generalized, `non-commutative', probability theory, in which classical
probability theory and quantum mechanics can be discussed in unison.
It consists of ordinary Hilbert space quantum theory,
with the emphasis moved towards operators on Hilbert space,
and the algebras which they generate.
The main objective of this course is to sketch the outlines of this framework,
and show its usefulness for information theory.

\subsection{Quantum Information}
In Shannon's (classical) information theory, a single unit,
the {\it bit}, serves to quantify all forms of information,
be it in print, in computer memory, CD-ROM or strings of DNA. Such a
single unit suffices, because different forms of information can be
converted into each other by copying, according to fixed `exchange
rates'.
The physical states of quantum systems, however, cannot be copied
into such `classical' information, but {\it can} be converted into
each other. This leads to a new unit of information: the {\it
qubit}.

Quantum Information theory studies the handling of this new form of
information by information-carrying channels.
We shall treat the basic properties of these channels,
and some impossibilities as well as new possibilities
connected with quantum information.
The impossibility of copying makes quantum information an ideal
means to establish secrecy \cite{Bruss}.

\subsection{Quantum Computing}
It was Richard Feynman who first thought of actually {\it employing}
the strangeness of quantum mechanics to do things that
would be impossible in a classical world.

The idea was developed in the 1980's and 1990's by David Deutsch,
Peter Shor, and many others into a flourishing branch of science
called `quantum computing': how to make quantummechanical systems
perform calculations more efficiently than ordinary computers can
do. This research is still in a predominantly theoretical stage: the
quantum computers actually built are as yet extremely primitive and
can by no means compete with even the simplest pocket calculator,
but expectations are high \cite{Kempe}.

\subsection{This Course}
We start with an introduction to quantum probability.
In Section \ref{sec:Bell}
we demonstrate the `strangeness' of quantum phenomena
by very simple polarization experiments,
culminating in Bell's famous inequality, tested in Aspect's experiment.
Bell's inequality is a statement in classical probability that is
violated in quantum probability and in reality.

Taking polarizers as our starting point,
in Sections \ref{sec:MathMod} and \ref{sec:QuantProb} we build up the new
probability theory in terms of algebras of operators on a Hilbert
space.
In Section \ref{sec:Operations}
operations on these algebras will be characterized,
and some aspects will be discussed in which they differ from classical
physical operations.
They are subject to certain strange limitations:
the impossibility of copying, of coding information into bits,
of jointly measuring incompatible observables,
of observation without perturbing the object
(Cf. Section \ref{sec:QuantImp}.)
But they also open up surprising possibilities:
entangling remote systems, teleportation of this
entanglement, sending two bits in a single qubit.
(Cf. Section \ref{sec:QuantNov}.)
We shall leave
the further luring perspectives to other courses in this School:
highly efficient algorithms for sorting, Fourier transformation and
factoring very large numbers \cite{Kempe}.


\section{Why Classical Probability does not Suffice}\label{sec:Bell}
(This section is based on \cite{KuMa}.)


\subsection{An Experiment with Polarizers}
To start with, we consider a simple experiment. In a beam of light
of a fixed color we put a pair of polarizing filters, each of which
can be rotated around the axis formed by the beam. As is well known,
the light falling through both filters changes in intensity when the
filters are rotated relative to each other. Starting from the
orientation where the resulting intensity is maximal, and rotating
one of the filters through an angle $\a$, the light intensity
decreases with $\a$, vanishing for $\a=\half\pi$. If we call the
intensity of the beam before the filters $I_0$, after the first
$I_1$, and after the second $I_2$, then $I_1=\half I_0$, (we assume
the original beam to be unpolarized), and
\begin{equation}\label{CosSquare}
     I_2=I_1\cos^2\a.
\end{equation}

\afbeelding{3cm}{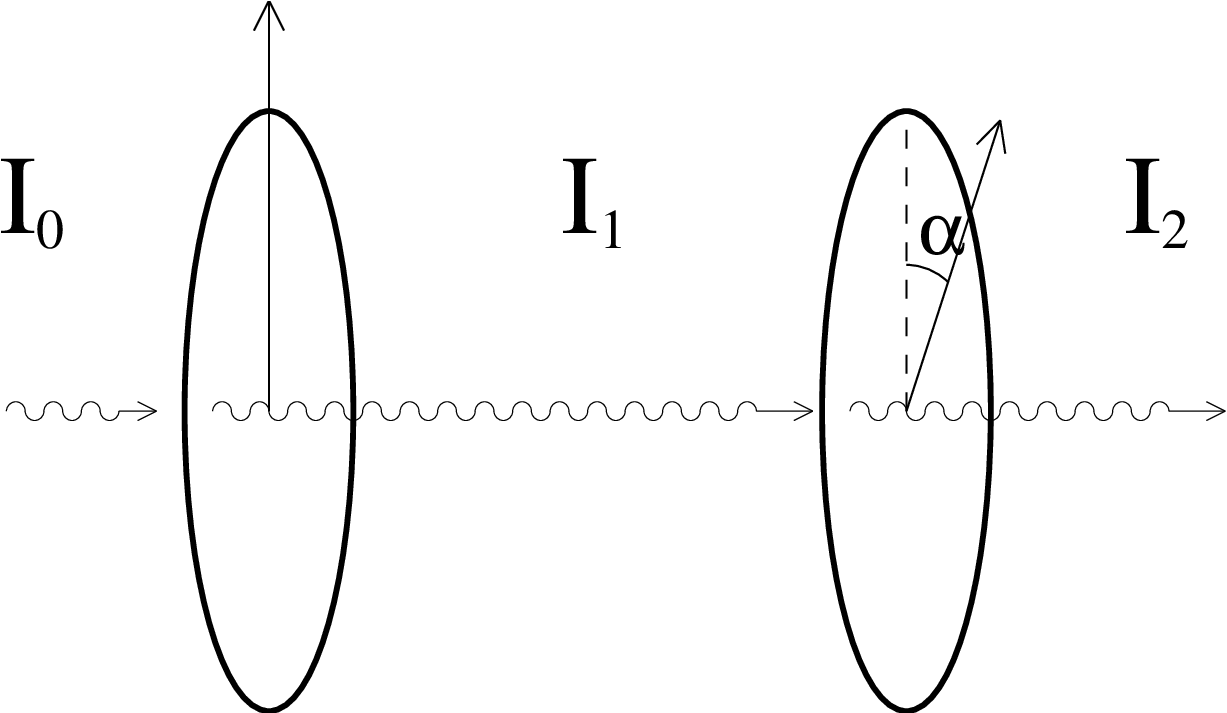}{Two polarizers in conjunction}{OptBench}

So far the phenomenon is described well by classical physics. During
the last century, however, it has been observed that for very low
intensities (monochromatic) light comes in small packages, which
were called {\it photons}, whose energy depends on the color, but
not on the total intensity.

\noindent So the intensity must be proportional to the {\it number}
of these photons, and formula (1) must be given a statistical
meaning: a photon passing through the first filter has a probability
$\cos^2\a$ to pass through the second. Formula (1) then only holds
on the average, for large numbers of photons.

\smallskip
Thinking along the lines of classical probability, we may associate
to a polarization filter in the direction $\a$ a random variable
$P_\a$, taking the value $P_\a(\o)=0$ if the photon $\o$ is absorbed
by the filter and $P_\a(\o)=1$ if it passes through. For two filters
in the directions $\a$ and $\b$ these random variables then should
be correlated as follows:
\begin{equation}\label{CosSquare}
\EE(P_\a P_\b)\ =\ \PP[P_\a=1\and P_\b=1]\ =\ \half\cos^2(\a-\b).
\end{equation}
Here we hit on a difficulty: the function on the right hand side is
not a possible correlation function! This can be seen as follows.
Take three polarizing filters, having polarization directions
$\a_1$, $\a_2$ and $\a_3$ respectively. Put them on the optical
bench in pairs. They should give rise to random variables $P_1$,
$P_2$ and $P_3$ satisfying
    $$\EE(P_i P_j)=\half\cos^2(\a_i-\a_j).$$

\begin{proposition} {\rm(Bell's 3 variable inequality)}
For any three 0-1-valued random variables $P_1$, $P_2$, and $P_3$ on a
probability space $(\O,\PP)$ the following inequality holds:
   $$\PP[P_1=1,P_3=0]\ \ \le\ \ \PP[P_1=1,P_2=0]+\PP[P_2=1,P_3=0].$$
\end{proposition}
\begin{proof}
\begin{eqnarray*}
\PP[P_1=1,P_3=0]&=&\PP[P_1=1,P_2=0,P_3=0]
                                    +\PP[P_1=1,P_2=1,P_3=0]\cr
                   &\le&\PP[P_1=1,P_2=0]+\PP[P_2=1,P_3=0].
                   \qquad\qquad\qquad\qed\cr
\end{eqnarray*}
\end{proof}
In our example, we have
\begin{eqnarray*}
   \PP[P_i=1,P_j=0] \ &=& \ \PP[P_i=1]-\PP[P_i=1,P_j=1]\cr
                        &=& \ \half-\half\cos^2(\a_i-\a_j)
                         = \half\sin^2(\a_i-\a_j).\cr
\end{eqnarray*}
Bell's inequality thus reads
   $$\half\sin^2(\a_1-\a_3)\le\half\sin^2(\a_1-\a_2)+\half\sin^2(\a_2-\a_3),$$
which is clearly violated for the choices $\a_1=0$, $\a_2={1\over6}\pi$
and $\a_3={1\over3}\pi$, where it says that
   $${3\over8}\le{1\over8}+{1\over8}.$$
This example suggests that classical probability cannot even
describe this simple experiment!

\medskip\noindent{\it Remark}

\noindent
The above calculation could be summarized as follows:
we are in fact looking for a family of 0-1-valued random variables
$(P_\a)_{0\le\a<\pi}$
with $\PP[P_\a=1]={1\over2}$, satisfying the requirement that
\begin{equation}\label{BellDistance}
\PP[P_\a\ne P_\b]=\sin^2(\a-\b).
\end{equation}
Now, on the space of 0-1-valued random variables on a probability space
the function $(X,Y)\mapsto\PP[X\ne Y]$ equals the $L^1$-distance of $X$
and $Y$:
   $$\PP[X\ne Y]=\int_\O|X(\o)-Y(\o)|\,\PP(d\o)=\norm{X-Y}_1.$$
On the other hand, the function $(\a,\b)\mapsto\sin^2(\a-\b)$ does
not satisfy the triangle inequality for a distance function on the
interval $[0,\pi)$. Therefore no family $(P_\a)_{0\le\a<\pi}$ exists
which meets the above requirement (\ref{BellDistance}).

\subsection{An improved experiment}
On closer inspection the above example is not very convincing.
Indeed, when two polarizers are arranged on the optical bench, why
should not the random variable for the second polarizer depend on
the angle of the first? The correlation in (\ref{CosSquare}) would
then read
   $$\EE(P_\a P_{\a,\b})\ =\ \PP[P_\a=1\and P_{\a,\b}=1]\ =\
   \half\cos^2(\a-\b),$$
which can easily be satisfied, and the whole refutation
collapses.

\noindent
So we should do a better experiment. We must let the filters act on
the photons without influence on each other. Maybe we can separate
them spatially?

\noindent Here a clever technique from quantum optics comes to our
aid. It is possible to build a device that produces {\it pairs} of
photons, such that the members of each pair move in opposite
directions and show opposite behavior towards parallel polarization
filters: if one passes the filter, then the other is surely
absorbed. The device contains Calcium atoms, which are excited by a
laser to a state they can only leave under emission of such a pair.

\afbeelding{3cm}{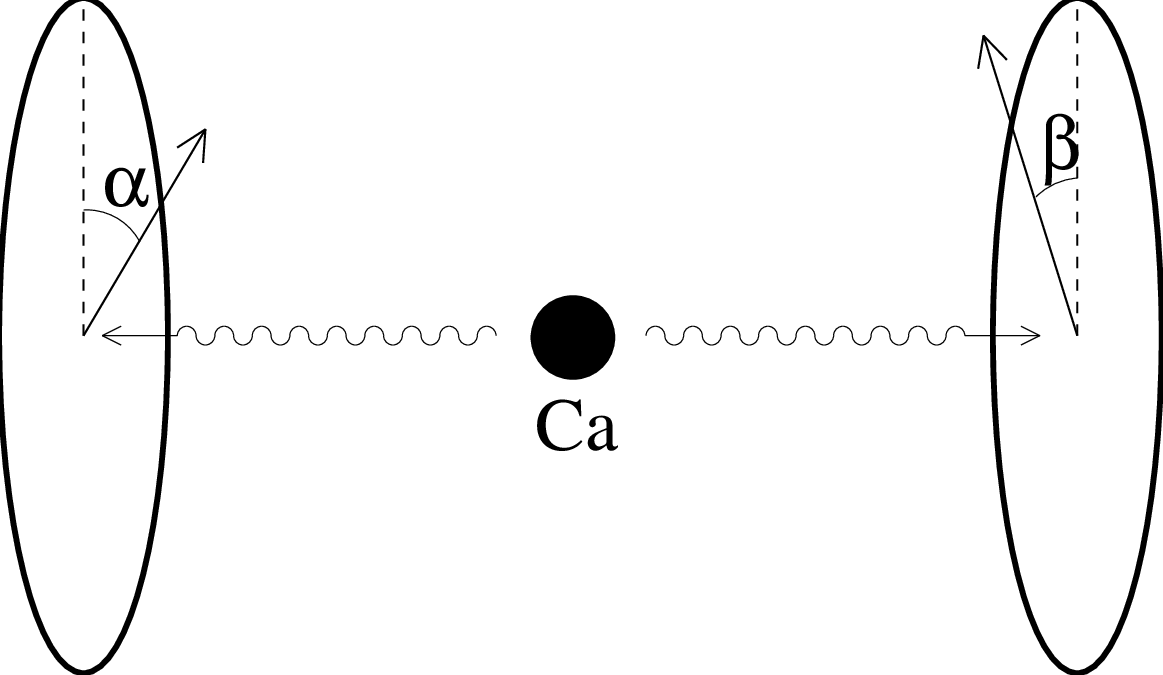}{Photon pair production}{Calcium}

\bigskip\noindent
With these photon pairs, the very same experiment can be performed,
but this time the polarizers are far apart, each one acting on its
own photon. The same correlations are measured, say first between
$P_{\a_1}$ on the left and $P_{\a_2}$ on the right, then between
$P_{\a_1}$ on the left and $P_{\a_3}$ on the right, and finally
between $P_{\a_2}$ on the left and $P_{\a_3}$ on the right. The same
outcomes are found, violating Bell's three variable inequality, thus
strengthening the case against classical probability.

\subsection{The decisive experiment}\label{decisive}
Advocates of classical probability could still find serious fault
with the argument given so far. Indeed, do we really {\it have to}
assume that we are measuring the same random variable $P_{\a_2}$ on
the right as later on the left? Is it really true that the
polarizations in these pairs are exactly opposite? There could exist
a probabilistic explanation of the phenomena without this
assumption.

\smallskip
So the argument has to be tightened still further. This brings us to
the experiment which was actually performed by A. Aspect in Orsay
(near Paris) in 1982 \cite{Aspect}. In this experiment a random
choice out of two different polarization measurements was performed
on each side of the pair-producing device, say in the direction
$\a_1$ or $\a_2$ on the left and in the direction $\b_1$ or $\b_2$
on the right, giving rise to {\it four} random variables
$P_1:=P(\a_1)$, $P_2:=P(\a_2)$ and $Q_1:=Q(\b_1)$, $Q_2:=Q(\b_2)$,
two of which are measured and compared at each trial.

\medskip
\begin{proposition} {\rm(Bell's 4 variable inequality)}
For any quadruple $P_1$, $P_2$, $Q_1$, and $Q_2$ of 0-1-valued
random variables on $(\Omega,\PP)$ the following inequality holds:
\begin{equation}\label{BellIneq}
   \PP[P_1=Q_1]\ \ \le\ \ \PP[P_1=Q_2]+\PP[Q_2=P_2]+\PP[P_2=Q_1].
\end{equation}
\end{proposition}
(In fact, by symmetry, neither of these four probabilities is larger
than the sum of the other three.)

\begin{proof}
It is easy to see that for all $\o$:
\begin{equation}\label{BellPointwise}
     P_1(\o)=Q_1(\o)\Longrightarrow
     P_1(\o)=Q_2(\o)\hbox{ or }Q_2(\o)=P_2(\o)\hbox{ or }P_2(\o)=Q_1(\o)\;.
     \quad\qed
\end{equation}
\end{proof}

\bigskip\noindent
Bell's 4-variable inequality can be viewed as a `quadrangle
inequality' with respect to the metric $(X,Y)\mapsto\norm{X-Y}_1$
on random variables $X$, $Y$.

\smallskip
On the other hand, quantum mechanics predicts (cf. Section \ref{MathAspect} below),
and the experiment of Aspect showed, that one has,
   $$\PP[P(\a)=Q(\b)=1]=\half\sin^2(\a-\b).$$
Similarly, $\PP[P(\a)=Q(\b)=0]=\half\sin^2(\a-\b)$.
Hence
   $$\PP[P(\a)=Q(\b)]=\sin^2(\a-\b).$$
So Bell's 4 variable inequality reads in this example:
   $$\sk11\le\sk12+\sk21+\sk22,$$
which is clearly violated for the choices
$\a_1=0$, $\a_2={\pi\over3}$, $\b_1={\pi\over2}$, and $\b_2={\pi\over6}$,
in which case it reads
   $$1\le{1\over4}+{1\over4}+{1\over4}.$$

\afbeelding{3cm}{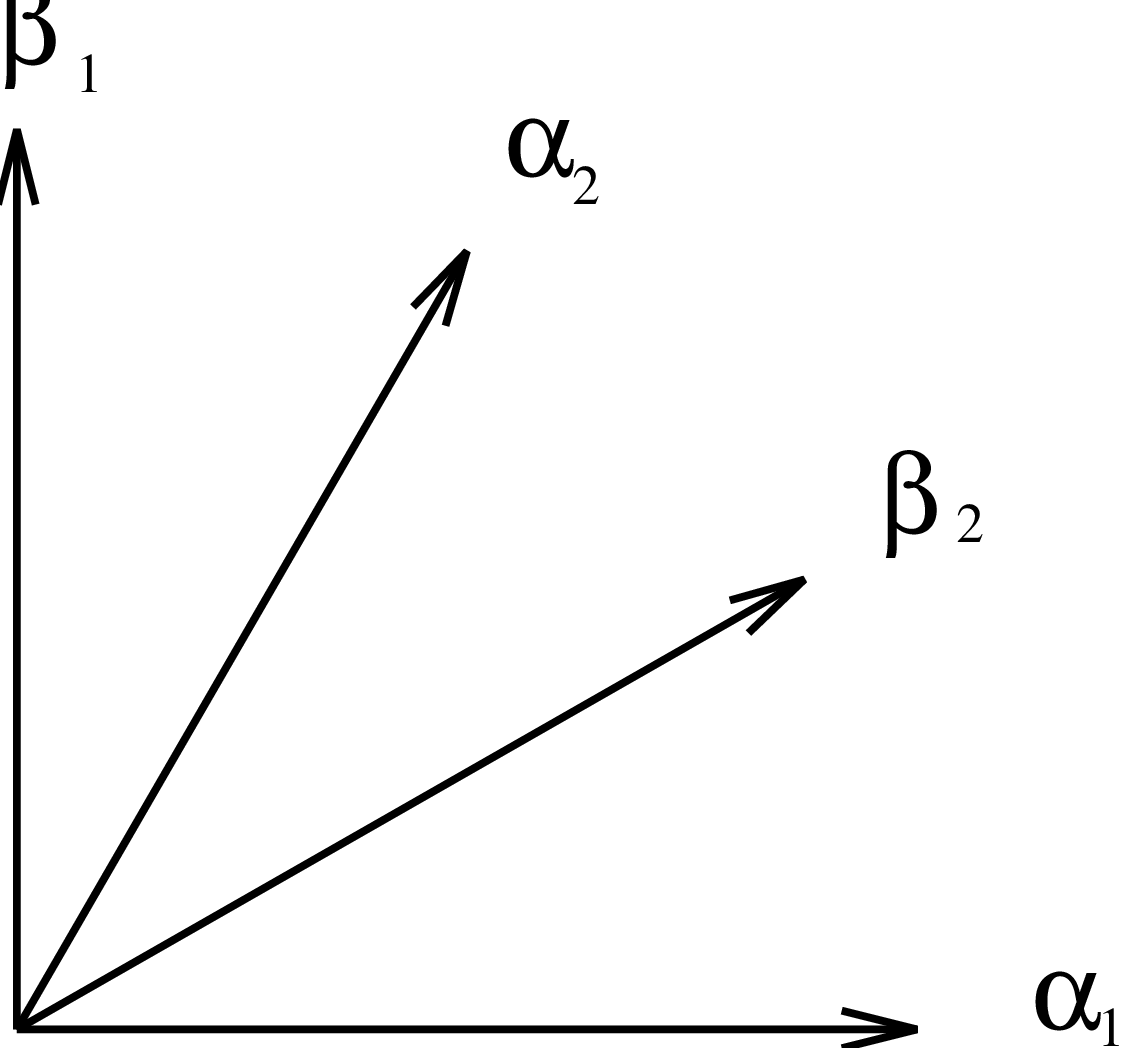}{Directions violating Bell's inequality}{settings}

\noindent Now we are finished: There does not exist, on any
classical probability space, a quadruple $P_1$, $P_2$, $Q_1$, and
$Q_2$ of random variables with the correlations measured in this
experiment.

\medskip\noindent{\it Discussion}
\begin{enumerate}

\item
A crucial assumption that goes into Bell's inequality, is that it
makes sense to compare (Cf. (\ref{BellPointwise})) the (possibly
random) reactions which a given photon {\it would} show to different
filters, including those it does not actually meet. This assumption
is called {\it realism}; it is made in all classical probabilistic
physical theories, but is abandoned in quantum mechanics.

\item
A second important assumption, necessary for the validity of Bell's
inequality, was mentioned before: the {\it outcome} on the right
(described by $Q(\b)$ for some $\b$) should not depend on the {\it
angle} $\a$ of the polarizer on the left. This assumption is called
`locality'. In order to justify this assumption, Aspect has made
considerable efforts. In his (third) experiment \cite{Aspect},
the choice of what
to measure on the left ($\a_1$ or $\a_2$) and on the right ($\b_1$
or $\b_2$) was made {\it during the flight of the photons}, so that
any influence which each of these choices might have on {\it the
outcome} on the opposite end would have to travel faster than light.
By the causality principle of Relativity Theory such influences are
excluded.

\item
The Orsay experiment refutes all imaginable physical theories which
are both {\it local} and {\it realistic}(Cf. 1 and 2 above).
Quantum mechanics is local, but not realistic.
Its great successes lead us to believe
that realism is fails in for the description of nature.
Some prefer to adhere to realism,
and so they must give up locality, and hence Einstein causality
\cite{Bohm, Nelson}.

\item
In our opinion, the phrase `quantum non-locality', which is often
heard in the context of Bell's inequalities, signals a
misconception.
It suggests giving up {\it both} realism {\it and} locality.
This is too much of a defeat, and unnecessary.
Quantum mechanics is local. But it describes
phenomena which {\it in a classical theory} could only be explained
using some action at a distance.

\end{enumerate}

\subsection{The Orsay experiment as a card game}
To illustrate the above refutation of local realism more vividly, we
shall present the experiment in the form of a card game. Nature can
win this game. Can you?

\afbeelding{6cm}{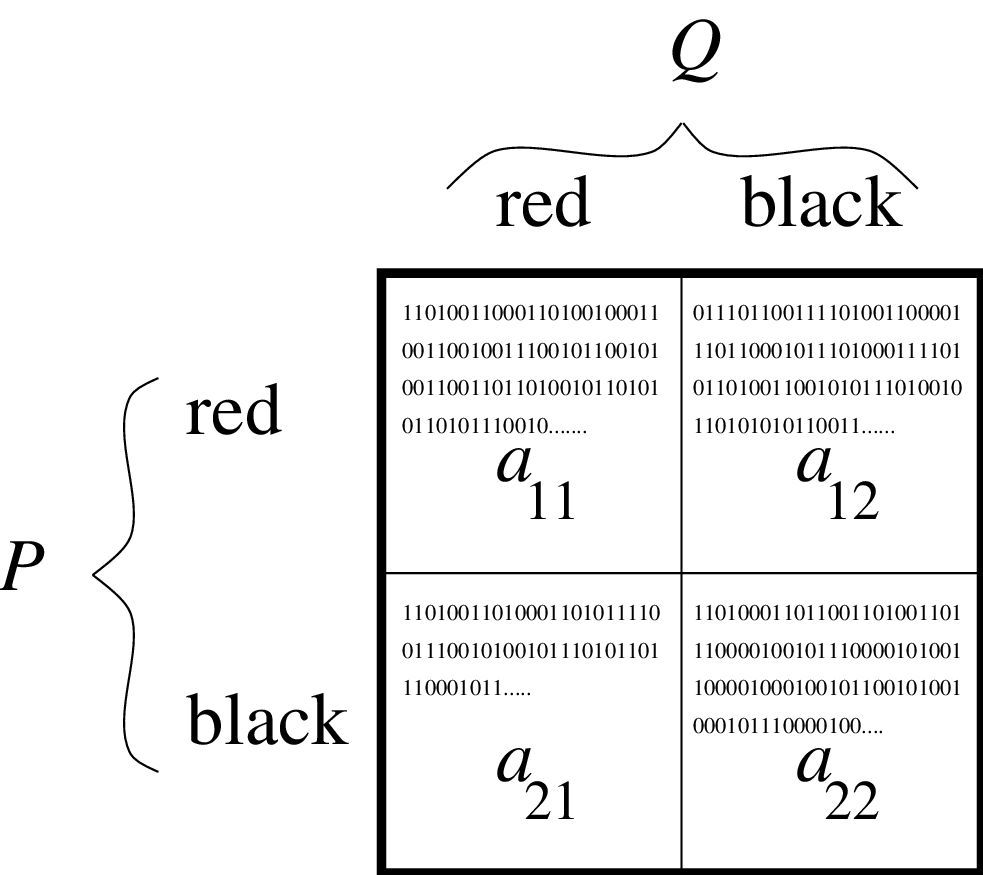}{Board for the Bell game}{game}

\bigskip\noindent
Two players, $P$ and $Q$, are sitting at a table.
They are cooperating to achieve a single goal.
There is an arbiter present to deal cards and to count points.
On the table there is a board consisting of four squares as drawn
in fig. 4. There are dice and an ordinary deck of playing cards.
The deck of cards is shuffled well.
(In fact we shall assume that the deck of cards is an infinite sequence
of independent cards, chosen fully at random.)
First the players are given some time to make agreements on the strategy
they are going to follow.
Then the game starts,
and from this moment on they are no longer allowed to communicate.
The following sequence of actions is then repeated many times.

\begin{enumerate}
\item
The dealer hands a card to $P$ and one to $Q$. Both look at their
own card, but not at the other one's. (The only feature of the card
that matters is its colour: red or black.)

\item
The dice are thrown.

\item
$P$ and $Q$ simultaneously say `yes' or `no',
according to their own choice.
They are free to make their answer depend on any information they possess,
such as the color of their own card, the agreements made in advance,
the numbers shown by the dice, the weather, the time, et cetera.

\item
The cards are laid out on the table. The pair of colors of the cards
determines one of the four squares on the board: these are labeled
(red,red), (red,black), (black,red) and (black,black).

\item
In the square so determined a 0 or a 1 is written:
a 0 when the answers of $P$
and $Q$ have been different, a 1 if they have been the same.
\end{enumerate}

\noindent
In the course of time, the squares on the board get filled with 0's
and 1's. The arbiter keeps track of the percentage of 1's in
proportion to the total number of bits in each square; we shall call
the time limits of these percentages as the game proceeds: $a_{11}$,
$a_{12}$, $a_{21}$, and $a_{22}$. The aim of the game, for both $P$
and $Q$, is to get $a_{11}$ larger than the sum of the other three
limiting percentages. So $P$ and $Q$ must try to give identical
anwers as often as they can when both their cards are red, but
different answers otherwise.

\medskip\noindent`PROPOSITION'.
{\rm (Bell's inequality for the game)}
{\it $P$ and $Q$ cannot win the game by classical means, namely:}
   $$a_{11}\le a_{12}+a_{21}+a_{22}.$$

\smallskip\noindent`{\it Proof\/}'.

\noindent The best $P$ and $Q$ can do, in order to win the game, is
to agree upon some (possibly random) strategy for each turn. For
instance, they may agree that $P$ will always say `yes' (i.e.,
$P_\red=P_\black=$`yes') and that $Q$ will answer the question `Is
my card red?' (i.e., $Q_\red=$ `yes' and $Q_\black=$`no'). This will
lead to a 1 in the (red, red) square or the (black, red) square or to
a 0 in one of the other two. So if the players repeat this strategy
indefinitely , on the long run they would get $a_{11}=a_{12}=1$ and
$a_{21}=a_{22}=0$, disappointingly satisfying Bell's inequality.

The above example is an extremal strategy. There are many (in fact,
sixteen) strategies like this. By the pointwise version
(\ref{BellPointwise}) of Bell's 4-variable inequality, none of these
sixteen extremal strategies wins the game. Inclusion of the
randomness coming from the dice yields a full polytope of random
strategies, having the above sixteen as its extremal points. But
since the inequalities are linear, this averaging procedure does not
help. This `proves' our `proposition'. Disbelievers are challenged
to find a winning strategy. \qed

\smallskip
Strangely enough, however, Nature does provide us with a strategy to win the
game, still essentially based on the $\cos^2$ law (\ref{CosSquare})
for photon absorption!
Instead of the dice, put a Calcium atom on the table.
When the cards have been dealt, $P$ and $Q$ put their
polarizers in the direction indicated by their cards.
If $P$ has a red card, then he chooses the direction $\a_1=0$ (cf. fig. 3).
If his card is black, then he chooses $\a_2={\pi\over3}$.
If $Q$ has a red card, then he chooses $\b_1={\pi\over2}$.
If his card is black, then he chooses $\b_2={\pi\over6}$.
No information on the colours of the cards needs to be exchanged.
When the Calcium atom has produced its photon pair,
each player looks whether his own photon passes his own polarizer,
and then says `yes' if it does, `no' if it does not.
On the long run they will get $a_{11}=1$, $a_{12}=a_{21}=
a_{22}={1\over4}$, and thus they win the game.

\smallskip So the Calcium atom, the quantummechanical die, makes
possible what could not be done with the classical die.

\section{Towards a Mathematical Model}
\label{sec:MathMod}
Coerced by the foregoing considerations, we give up trying
to make a classical probabilistic model in order to explain
polarization experiments.
Instead, we take these experiments as a paradigm for an alternative type
of `quantum' probability, to be developed now.

\subsection{A mathematical description of polarization}
We have discussed (linear) polarization of a light beam. This is
completely characterized by a direction in the plane perpendicular
to the light beam. So we simply describe states of polarization by
different directions in a two-dimensional real plane $\RR^2$, or
equivalently by unit vectors $\psi \in \RR^2$, $\norm{\psi}=1$,
pointing in this direction. Actually, since we cannot distinguish
between two states which differ by a rotation of $\pi$, we shall
describe states of polarization by one-dimensional subspaces of
$\RR^2$. Given two directions of polarization with an angle $\a$
between them, spanned by two unit vectors $\psi,\theta \in \RR^2$,
the probability to find polarization $\th$ when a photon is in the
state $\psi$, can be expressed as
                $$ \cos^2 \a = \inp\psi\theta^2$$
where $\inp\psi\theta$ denotes the scalar product between $\psi$ and $\theta$.

In the mathematical model we should distinguish between the physical
state of polarization of a photon on the one hand and the filter on
the other hand, i.e., the 0-1-valued random variable which asks,
whether a photon is polarized in a certain direction. This can be
done by identifying the random variable with the orthogonal
projection $P$ onto the one-dimensional subspace. We can then write
             $$ \cos^2 \a \ = \inp\psi\theta^2 \ = \inp\psi{P\psi} \ .$$
Since $P$ is 0-1-valued, (a photon passes or is absorbed), this
probability is equal to the expectation of this random variable:
     $$\inp\psi{P\psi} \ = \ \EE(P)\ .$$


\subsection{The full truth about polarization: the qubit}
\label{qubitphoton}
In the foregoing description of polarization things were presented
somewhat simpler than they are: we considered only linear
polarization, thus disregarding circular polarization. The full description
of polarization leads to the quantum mechanics of a 2-level system
or {\it qubit}:

\correspondence{State of polarization of a photon}
               {one-dimensional subspace of $\CC^2$, de\-scri\-bed
                by a unit vector $\psi$ spanning this subspace
                (and determined only up to a phase).}

\correspondence{Polarization filter or generalized {{0-1}}-valued random
                variable}
               {orthogonal projection $P$ onto a complex one-dimensional
                subspace.}

\noindent (Also for left- or right-circular polarization there
exist physical filters.)

\correspondence{Probability for a photon, described by $\psi$, to pass through
                a filter, described by $P$}
               {$\inp\psi{P\psi}$\ .}

\smallskip
The set of all states is conveniently parametrized by the unit vectors
of the form
\begin{equation}\label{unitvectorC2}
(\cos \a, e^{i\phi}\sin \a) \in \CC^2 \ ,\quad
{-\pi \over 2} \le \a \le {\pi \over 2}, \quad 0 \le \phi \le \pi\ .
\end{equation}


\subsection{Finite dimensional models}
The mathematical model that is used by quantum mechanics is the
straightforward generalization of the above description. In order to
keep things simple, in this course we restrict ourselves to the
quantum mechanics in finite dimension. This generalizes the
probability theory of systems with only finitely many states. As in
classical probability, the generalization to systems with a
countably infinite number of states or a continuum of states is
analytically more involved.

The model is as follows:
\smallskip
{\it States} correspond to one-dimensional subspaces of $\CC^n$, where
the dimension $n$ is determined by the model. Again, a state is described
conveniently by some unit vector spanning this subspace.

\smallskip
{\it 0-1-valued random variables} or {\it events}
are described by orthogonal projections onto linear subspaces of $\CC^n$.
Here also projections onto higher dimensional subspaces
make sense.

\smallskip
The {\it probability} that a measurement of a random variable $P$ on a
system in a state $\psi$ gives the value $1$ is given by
$\inp\psi{P\psi}$.

\smallskip
Note that we do not assume that every orthogonal
projection corresponds to a meaningful random variable.
Specification of random variables is part of the description
of the mathematical model for a given system.
In a truly quantum mechanical situation, typically
all projections are used. In contrast to this,
classical probability is obtained by allowing only very few
projections, as follows.

\subsection{Finite classical models}

\noindent
A finite probability space is usually described by a finite set
$\Omega = \{\o_1, \ \dots,\ \o_n\}$ and a probability distribution
$(p_1, \ \dots, \ p_n)$, $0 \le p_i \le 1$, $\sum_i p_i = 1$,
such that the probability for $\o_i$ is $p_i$. A 0-1-valued random
variable is a 0-1-valued function on $\Omega$, i.e., a characteristic
function $\chi_A$ of some subset $A \subseteq \Omega$.
In order to describe such a system in our model, we think of $\CC^n$ as the
space of complex valued functions on $\Omega$, and use the functions
$\delta_i$ with $\delta_i(\o_j) \ = \ \delta_{i,j}$ as basis.
The states of the system, i.e., the points $\o_i$ of $\Omega$, are now
represented by the unit vectors $\delta_i$, $1 \le j \le n$.
The random variable $\chi_A$ is identified with the orthogonal
projection $P_A$ onto the linear span of the vectors
$\{\delta_i: \ \o_i \in A\}$. In our basis $\chi_A$ becomes a diagonal
matrix with a $1$ at the $i$-th place of the diagonal if $\o_i \in A$,
and a $0$ otherwise. It is obvious that $\o_i \in A$ if and only if
$\chi_A(\o_i)=1$ if and only if $\inp{\delta_i}{P_A\delta_i}\ =\ 1$.

Conversely, any set of pairwise commuting
projections on $\CC^n$ can be diagonalized simultaneously and
thus have an interpretation as a set of classical 0-1-valued random
variables. Therefore:

\medbreak
\centerline{{\it  Classical probability corresponds to sets of
pairwise commuting projections.}}
\medbreak

\subsection{Mixed states}

In the above sketch of quantum probability an important point
is still missing:
How can we describe
a situation where a photon has one polarization with some probability
$q$ and in another with probability $1-q$?
Since states must play the role of probability distributions,
this combination should be expressed as a single state of the photon.

In general, if $P$ is any 0-1-valued (quantum) random variable and
$\psi_1, \dots, \psi_k$ are arbitrary quantum states, each occuring with
a probability $p_i$, $1 \le i \le k$, $\sum_i p_i = 1$,
then the probability that a measurement of $P$ gives $1$ is clearly
given by
$$ \sum_i p_i\inp{\psi_i}{P\psi_i}\ \ .$$
A more convenient description of mixed states is obtained as follows.

For a unit vector $\psi \in \CC^n$ denote by $\rho_\psi$ the orthogonal
projection onto the one-dimensional subspace generated by $\psi$.
In the physics literature, $\rho_\psi$ is often denoted by
$\ket\psi\bra\psi$. By $tr$ denote the trace on the
$n \times n$-matrices, summing up the diagonal entries of such a matrix.
Then one obtains
            $$\inp\psi{P\psi}\ = \ tr(\rho_\psi\cdot P) \ .$$
Hence
$$\sum_i p_i\inp{\psi_i}{P\psi_i}\
=\ tr(\sum_i p_i\rho_{\psi_i} \cdot P)\ =\ tr(\rho \cdot P) \ ,$$
where $\rho := \sum_i p_i \rho_{\psi_i}$.

Being a convex combination of 1-dimensional projections, $\rho$ obviously
is a positive (i.e., self-adjoint positive semidefinite) $n\times n$-matrix
with $tr(\rho) = 1$. Conversely, from diagonalizing positive matrices it is
clear that any such positive matrix $\rho$ with $tr(\rho) = 1$ can be
written as a convex combination of 1-dimensional projections. The set of
these matrices forms a closed (even compact) convex set, and its extreme
points are precisely the 1-dimensional projections which in turn
correspond to pure states, represented also by unit vectors.
Therefore it is this class of so-called {\it density matrices}
which represents mixed states. 

Thus, a general mixed state is described by a density matrix $\rho$ and
the probability for an observation of $P$ to yield the value $1$ is
given by $tr(\rho\cdot P)$.

\noindent{\it Remarks}

\begin{enumerate}

\item
The decomposition of a density matrix $\rho$ into a convex
combination of 1-dimensional projections is by no means unique. This
point will be further elaborated in Proposition \ref{staterepr} of
Section \ref{subsec:uniqueness}. So the compact convex set of
density matrices is not a simplex at all. Indeed, on $\CC^2$ it can
be affinely identified with a full ball in $\RR^3$, by taking in
$\RR^3$ the convex hull of the sphere that was described above.

\item
In classical probability the convex set of mixed states is the
simplex of all
probability distributions. In our picture, if we insist on decomposing
a mixed state given by $\rho = \sum_i p_i P_{\delta_i}$ into a convex
combination of pure states (within the convex hull of
$\{P_{\delta_i} : 1 \le i \le n\}$ which is a simplex), then it becomes
unique.

\item
Physically, a state $\rho$ is completely described by all of its
values $tr(\rho \cdot P)$, where $P$ runs through the random variables
of the model. Thus, if we consider only subsets of projections, then two
different density matrices can represent the same physical state of the
system. As a drastic example, consider the classical system
$\Omega = \{\o_1, \ \dots,\ \o_n\}$ with equidistribution, i.e.,
$p_i(\o_i) = {1 \over n }$, leading to the density matrix
$\rho = \sum_i {1\over n} P_{\delta_i} = {1\over n} \cdot \11$.
On the other hand, with the unit vector
$\psi = ({1\over \sqrt n}, \ \dots \ ,{1\over \sqrt n})\in \CC^n$,
we obtain for any subset $A \subseteq \Omega$:
$tr(\rho\cdot P_A) = {1\over n}\cdot \vert A\vert = \ {\inp\psi{P_A\psi}}$.
Therefore, on the random variables $\{P_A:A \subseteq \Omega\}$,
the rank-one-density matrix $P_\psi$ represents the same state as
the densitiy matrix ${1\over n}\cdot \11$.
Note, however, that $P_\psi$ is not in the
convex hull of $\{P_{\delta_i} : 1 \le i \le n\}$.

\end{enumerate}

\subsection{The mathematical model of Aspect's experiment}\label{MathAspect}
As an illustration, we shall now explain the photon correlation in the
Orsay experiment, given by the $\cos^2$-law. Note that here we cannot
simply refer to the basic $\cos^2$-law of quantum probability, since
the filters are acting on two different photons.

The polarization of a pair of photons is described by a unit vector
in the tensor product ${\CC^2 \otimes \CC^2}\allowbreak = \CC^4$,
where we use the basis
\begin{eqnarray*}
(1,0,0,0) &= e_1 \otimes e_1 =: e_{11}, \cr
(0,1,0,0) &= e_1 \otimes e_2 =: e_{12}, \cr
(0,0,1,0) &= e_2 \otimes e_1 =: e_{21}, \cr
(0,0,0,1) &= e_2 \otimes e_2 =: e_{22}, \cr
\end{eqnarray*}
with $e_1 = (1,0)\in \CC^2$ and $e_2 = (0,1) \in \CC^2$.
For example, in the pure state $e_{12}$ the left-hand photon is
vertically polarized and the right-hand photon horizontally. As it
turns out, the state of the pair of photons as produced by the Calcium
atom is described by the state
    $$\psi = {1\over \sqrt 2}(e_{12} - e_{21}).$$
Now, the filters $P(\a)$ on the left and $Q(\b)$ on the right,
introduced in Section \ref{decisive}, are represented by
two-dimensional projection operators on $\CC^4$, which are the
``2-right amplification" and the ``2-left-amplification" of the
polarization matrix
  $$\pmatrix{\cos^2\a &\cos\a\sin\a \cr \cos\a\sin\a &\sin^2\a\cr},$$
namely

\begin{eqnarray*}
P(\a) &= \pmatrix{\cos^2\a &\cos\a\sin\a \cr \cos\a\sin\a &\sin^2\a\cr} \otimes
         \pmatrix{1 & 0 \cr 0 &1\cr} \cr
      &{\ }\cr
      &=\pmatrix{\cos^2\a     &0            &\cos\a\sin\a  &0            \cr
                 0            &\cos^2\a     &0             &\cos\a\sin\a \cr
                 \cos\a\sin\a &0            &\sin^2\a      &0            \cr
                 0            &\cos\a\sin\a &0             &\sin^2\a     \cr}
\cr
\end{eqnarray*}

\begin{eqnarray*}
Q(\b) &= \pmatrix{1 & 0 \cr 0 &1\cr} \otimes
         \pmatrix{\cos^2\b &\cos\b\sin\b \cr \cos\b\sin\b&\sin^2\b\cr}\cr
      &{\ }\cr
      &=\pmatrix{\cos^2\b     &\cos\b\sin\b &0             &0            \cr
                 \cos\b\sin\b &\sin^2\b     &0             &0            \cr
                 0            &0            &\cos^2\b      &\cos\b\sin\b \cr
                 0            &0            &\cos\b\sin\b  &\sin^2\b     \cr}
\ .\cr
\end{eqnarray*}

We note that $P(\a)$ and $Q(\b)$ are commuting projections for fixed
$\a$ and $\b$. It follows that $P(\a)Q(\b)$ is again a projection,
as well as the products ${P(\a)(\11-Q(\b))}$, $(\11-P(\a))Q(\b)$, and
$(\11-P(\a))(\11-Q(\b))$. So we obtain the description of a classical
probability space with four states, to be interpreted as
\begin{eqnarray*}
\hbox{(``left photon passes"}
                   &\hbox{, ``right photon passes"),}\cr
              \hbox{(``left photon passes"}
                   &\hbox{, ``right photon is absorbed")},\cr
              \hbox{(``left photon is absorbed"}
                   &\hbox{, ``right photon passes"),}\cr
       \hbox{(``left photon is absorbed"}
                   &\hbox{, ``right photon is absorbed").}\cr
\end{eqnarray*}
The probabilities of these four events are found by the actions on
$\psi={1\over\sqrt{2}}(e_{12}-e_{21})={1\over2}(0,1,-1,0)$
of the four projections.
In particular, the probability that both photons pass is given by

\noindent
{
\let\dis\displaystyle
\def\displaystyle{\scriptscriptstyle} \def\d{\displaystyle}
\begin{eqnarray*}
\dis
&\dis<\!\psi,P(\a)Q(\b)\psi\!>               \cr
&\dis= {1 \over 2}(0,1,-1,0)\times \cr
&\quad\; \dis\times
\pmatrix{\d\cos^2\a\cos^2\b            &\d\cos^2\a\cos\b\sin\b
                 &\d\cos\a\sin\a\cos^2\b       &\d\cos\a\sin\a\cos\b\sin\b \cr
         \d\cos^2\a\cos\b\sin\b        &\d\cos^2\a\sin^2\b
                 &\d\cos\a\sin\a\cos\b\sin\b   &\d\cos\a\sin\a\sin^2\b \cr
         \d\cos\a\sin\a\cos^2\b        &\d\cos\a\sin\a\cos\b\sin\b
                 &\d\sin^2\a\cos^2\b           &\d\sin^2\a\cos\b\sin\b     \cr
         \d\cos\a\sin\a\cos\b\sin\b    &\d\cos\a\sin\a\sin^2\b
                 &\d\sin^2\a\cos\b\sin\b       &\d\sin^2\a\sin^2\b         \cr
        }
 \pmatrix{\d 0\cr \d 1\cr \d -1\cr \d 0\cr} \cr
&\dis={1\over 2}(\cos^2\a\sin^2\b + \sin^2\a\cos^2\b - 2 \cos\a\sin\a\cos\b\sin\b)\cr
&\dis={1\over 2}(\cos\a\sin\b - \sin\a\cos\b)^2 \cr
&\dis={1\over 2}\sin^2(\a-\b) \ \ .\cr
\end{eqnarray*}
}

\section{Quantum Probability}
\label{sec:QuantProb}
In classical probability a model --- or {\it probability space} ---
is determined by giving a set $\O$ of outcomes $\o$,
by specifying what subsets $S\subset\O$ are to be considered as {\it events},
and by associating a {\it probability} $\PP(S)$ to each of these events.
Requirements:
the events must form a $\s$-algebra, the probability measure $\PP$
must be $\s$-additive, and normalized, i.e. $\PP(\O)=1$.

\pn
In quantum probability we must loosen this scheme somewhat.
We must give up the set $\O$ of sample points:
a point $\o\in\O$ in a classical model decides about
the occurrence or non-occurrence of all events simultaneously,
and this we abandon.
Following our polarization example of Section \ref{sec:Bell}
we take as {\it events}
certain {\it closed subspaces} of a {\it Hilbert space},
or, equivalently, a set of {\it projections}.
To all these projections we associate probabilities.

\sn
Requirements:
\begin{enumerate}
\item
The set of $\E$ of all events of a quantum model must be the set of projections
in some {\it $*$-algebra $\A$} of operators on $\H$.

\item
The probability function $\PP:\E\to[0,1]$ must be $\s$-additive.

\end{enumerate}
According to a theorem of Gleason, for $\dim(\H)\ge3$ this implies
that the probabilities are given by a {\it state} $\ph$ on $\A$:
  $$\PP(E)=\ph(E),\qquad (E\in\A \hbox{ a projection})\;.$$
In this section we shall work out
 the above notions in some detail.

\subsection{$*$-algebras of operators and states}

\noindent
A {\it Hilbert space} is a complex linear space $\H$
with a sesquilinear function
   $$\H\times\H\to\CC:\quad(\psi,\chi)\mapsto\inp\psi\chi\;,$$
the {\it inner product}.
For the defining properties of the inner product and the main facts about
Hilbert spaces
we refer to the contribution of D\'enes Petz to this volume \cite{Petz}.

\noindent
Let $\H$ be a finite-dimensional Hilbert space.
By an {\it operator} on $\H$ we mean a linear map $A:\H\to\H$.
Operators can be added and multiplied in the natural way.
By the {\it adjoint} of an operator $A$
we mean the unique operator $A^*$ on $\H$
satisfying
   $$\forall_{\psi,\th\in\H}:\quad \inp{A^*\psi}{\th}=\inp{\psi}{A\th}\;.$$
The {\it norm} of an operator $A$ is defined by
   $$\norm A:=\sup\set{\norm{A\psi}}{\psi\in\H,\norm\psi=1}\;.$$
It has the property
   $$\norm{A^*A}=\norms A\;.$$

\smallskip\noindent\emph{Exercise:}
Prove this!

\sn
By a {\it (unital) $*$-algebra of operators on $\H$}
we mean a subspace $\A$ of the
space of all linear maps $A:\H\to\H$ such that $\one\in\A$ and
   $$A,B\in\A \Implies \l A,\; A+B,\; A\cdot B,\; A^*\in\A\;.$$
By a {\it state} on $\A$ we mean a linear functional $\ph:\A\to\CC$
satisfying

\begin{enumerate}
\item
$\forall_{A\in\A}:\quad\ph(A^*A)\ge0$,

\item
$\ph(\one)=1$.
\end{enumerate}

We shall call a pair $(\A,\ph)$ of the above kind
a {\it quantum probability space}.

\bigskip\noindent
{\it Examples}

\begin{enumerate}
\item
Let $P_1$, $P_2$, $\ldots$, $P_k$ be mutually orthogonal projections
on $\H$ with sum $\one$.
Then their linear span
   $$\A:=\set{\sok j \l_j P_j}{\l_1,\ldots,\l_k\in\CC}\;.$$
forms a unital $*$-algebra of operators on $\H$.
This is basically the classical model of Section 2.4.:
$\A$ is isomorphic to $\C(\O)$,
the algebra of all complex functions on the finite set $\O=\{1,\ldots,k\}$.
If $\psi$ is some vector in $\H$ of unit length,
it determines a state $\ph$ by:
   $$\ph(A):=\inp{\psi}{A\psi}\;.$$
The probabilities of this classical model are $p_j:=\ph(P_j)=\norms{P_j\psi}$.
Note that there are many $\psi$'s, and even more density matrices $\rho$
(see Section 2.4.) determining the same state $\ph$ on $\A$.

\item
Let $\A$ be the $*$-algebra $M_n$ of all complex
$n\times n$ matrices.
Let $\ph(A):=\tr(\rho A)$ with $\rho\ge0$ and $\tr(\rho)=1$,
as introduced in Section 2.4.

\noindent
The state $\ph$ is called a {\it pure} state if $\rho=\ket\psi\bra\psi$
for some unit vector $\psi\in\H$.

\noindent
The qubit of Section 2.2 corresponds to the case $n=2$.

\noindent
The most general way of representing $M_n$ on a
(finite dimensional) Hilbert space is:
   $$\H=\CC^m\ten\CC^n\quad(m\ge1);\qquad \A=\set{\one\ten A}{A\in M_n}\;.$$

\item
Let $k$, $n_1,\ldots,n_k$, $m_1,\ldots,m_k$ be natural numbers,
and let the Hilbert space $\H$ be given by
   $$\H:=\left(\CC^{m_1}\ten\CC^{n_1}\right)\oplus
         \left(\CC^{m_2}\ten\CC^{n_2}\right)\oplus\cdots\oplus
         \left(\CC^{m_k}\ten\CC^{n_k}\right)\;.$$
Let \A be the $*$-algebra given by
   $$\A:=\set{(\one\ten A_1)\oplus\cdots\oplus(\one\ten A_k)}
             {A_j\in M_{n_j}\;\;\for\;\; j=1,\ldots,k}\;.$$
Let $\psi=\psi_1\oplus\ldots\oplus\psi_k$ be a unit vector in $\H$ and
   $$\ph(A):=\inp{\psi}{A\psi}=\sok j\inp{\psi_j}{A_j\psi_j}\;.$$
If $m_j\ge n_j\forall_j$ then every state on $\A$ is of the above form.
Otherwise, density matrices may be needed.

\end{enumerate}

In finite dimension
Example 1 is the only commutative possibility,
Example 2 is the `purely quantummechanical' situation,
and Example 3 is the most general case.

\begin{theorem}\label{Gelfand}
Every commutative $*$-algebra of operators on a finite-dimensional Hilbert
space is isomorphic to $\C(\O)$ for some finite $\O$.
\end{theorem}

This is the finite-dimensional version of Gel'fand's theorem on
commutative C*-algebra's.

\begin{proof}
Since the operators in $\A$ all commute,
there exists an orthonormal basis $e_1,\ldots,e_n$ in $\H$
on which they are all represented by diagonal matrices.
Then the states $\o_j:A\mapsto\inp{e_j}{Ae_j}$ are multiplicative:
   $$\o_j(AB)=\inp{e_j}{ABe_j}=\son i\inp{e_j}{Ae_i}\inp{e_i}{Be_j}
             =\inp{e_j}{Ae_j}\inp{e_j}{Be_j}
             =\o_j(A)\o_j(B)\;.$$
These states need not all be different;
let $\O:=(\o_{j_1},\ldots,\o_{j_k})$ be
a maximal set of different ones.
Then the map
   $$\iota:\A\to\C(\O):\iota(A)(\o):=\o(A)$$
is an isomorphism.
The projections of Example 1 are found back as
the operators $P_\o:=\iota^{-1}(\delta_\o)$.
\qed
\end{proof}

\noindent\emph{Exercise:}
Check that the map $\iota$ defined above is indeed an isomorphism
of $*$-algebras.

\begin{definition}
By the {\it commutant} of a set $\S$ of operators on $\H$
we mean the $*$-algebra
   $$\S':=\set{B:\H\to\H\hbox{ linear }}{\forall_{A\in\S}:AB=BA}\;.$$
The algebra generated by $\one$ and $\S$ we denote by $\alg(\S)$.
The {\it center} of a $*$-algebra $\A$ is the (commutative) $*$-algebra $\Z$
given by
   $$\Z:=\A\cap\A'\;.$$
\end{definition}

\noindent\emph{Exercise:} Find the center of $\A$ in each of the
examples 1, 2 and 3 above.

\begin{theorem}
(Double commutant theorem.)
Let $\S$ be a set of operators on a finite dimensional Hilbert space $\H$,
such that $X\in\S\implies X^*\in\S$.
Then
   $$\alg(\S)=\S''\;.$$
\end{theorem}

\begin{proof}
Clearly $\S\subset\S''$,
and since $\S''$ is a $*$-algebra,
we have $\alg(\S)\subset\S''$.
We shall now prove the converse inclusion.
Let $B\in\S''$, and let $\A:=\alg(\S)$.
We must show that $B\in\A$.

\sn\emph{Step 1:}
Choose $\psi\in\H$, and let $P$ be the orthogonal projection onto
$\A\psi$.
Then for all $X\in\S$ and $A\in\A$:
   $$XPA\psi=XA\psi\in\A\psi\Implies XPA\psi=PXA\psi\;.$$
So $XP$ and $PX$ coincide on the space ${\A\psi}$.
But if $\th\perp{\A\psi}$, then $P\th=0$ and for all $A\in\A$:
   $$\inp{X\th}{A\psi}=\inp{\th}{X^*A\psi}=0\;,$$
so $X\th\perp{\A\psi}$ as well.
Hence $PX\th=0=XP\th$, and the operators $XP$ and $PX$ also coincide on the
orthogonal complement of ${\A\psi}$.
We conclude that $XP=PX$, i.e. $P\in\S'$.
But then we also have $BP=PB$, since $B\in\S''$. So
   $$B\psi=BP\psi=PB\psi\in{\A\psi}\;,$$
and $B\psi$ is of the form $A\psi$ for some $A\in\A$.

\sn\emph{Step 2:}
But this is not sufficient:
we must show that $B\psi=A\psi$ for all $\psi$ in a basis for $\H$.

\noindent
So choose a basis $\tup\psi_n$ of $\H$.
We define:
\begin{eqnarray*}
\Ht &:=&\H\oplus\H\oplus\cdots\oplus\H=\CC^n\ten\H\;,\cr
\At&:=&\set{A\oplus A\oplus\cdots\oplus A}{A\in\A}=\A\ten\one\;,\cr
\pst&:=&\psi_1\oplus\psi_2\oplus\cdots\oplus\psi_n\;.\cr
\end{eqnarray*}
Then $(\At)'=(\A\ten\one)'=\A'\ten M_n$ and
$(\At)''=(\A'\ten M_n)'=\A''\ten\one$.
So $B\ten\one\in(\At)''$.
By step 1 we find an element $\widetilde A$ of $\At$,
such that
   $$\widetilde A\pst=(B\ten\one)\pst\;.$$
But $\widetilde A\in\At$ must be of the form $A\ten\one$ with $A\in\A$, so
   $$A\psi_1\oplus\cdots\oplus A\psi_n=B\psi_1\oplus\cdots\oplus B\psi_n\;.$$
This implies that $A=B$, hence $B\in\A$.
\qed
\end{proof}

\sn\emph{Exercise:} Find the algebra generated by $\one$ and the
matrix
  $$\pmatrix{0&1&0\cr 1&0&0\cr 0&0&0\cr}\;.$$

\sn
We give the following proposition without proof.
It characterizes the situation of Example 2.

\begin{proposition}\label{factor}
If the center of $\A$ contains only multiples of $\one$,
then $\H$ and $\A$ must be of the form
   $$\H=\CC^m\ten\CC^n, \with \A=\set{\one\ten A}{A\in M_n}\;.$$
\end{proposition}

\begin{proposition}
\label{GenStarAlg}
Let $\H$ be a finite-dimensional Hilbert space.
Then every $*$-algebra of operators on $\H$ can be written in the
form of Example 3 above.
\end{proposition}

\begin{proof}
The center $\A\cap\A'$ is an abelian $*$-algebra, so Theorem
\ref{Gelfand} applies, giving a set of projections $P_j$,
$j=1,\ldots,k$. Then it is not difficult to show that the unital
$*$-algebras $P_j\A P_j$ on the Hilbert subspaces $P_j\H$ satisfy
the condition of Proposition \ref{factor}. The statement follows.
\qed
\end{proof}


\subsection{The qubit}

\noindent
The simplest non-commutative $*$-algebra is $M_2$,
the algebra of all $2\times2$ matrices with complex entries.
And the simplest state on $M_2$ is $\half\tr$,
the quantum analogue of a fair coin.

\noindent
The events in this probability space are the orthogonal projections
in $M_2$: the complex $2\times2$ matrices $E$ satisfying
   $$E^2=E=E^*\;.$$
Let us see what these projections look like.
Since $E$ is self-adjoint,
it must have two real eigenvalues,
and since $E^2=E$ these must both be 0 or 1.
So we have three possibilities.

\begin{itemize}
\item
Both are 0; i.e. $E=0$.

\item
One of them is 0 and the other is 1.

\item
Both are 1; i.e. $E=\one$.
\end{itemize}

In the second case,
$E$ is a one-dimensional projection satisfying
   $$\tr E=0+1=1\and\det E=0\cdot1=0\;.$$
As $E^*=E$ and $\tr E=1$ we may write
\begin{equation}\label{projunitvectorR3}
   E=E(x,y,z)=\half\left(\matrix{1+z&x-iy\cr x+iy&1-z\cr}\right)\;.
\end{equation}
Then $\det E=0$ implies that
   $$\quarter((1-z^2)-(x^2+y^2))=0\Implies x^2+y^2+z^2=1\;.$$
So the one-dimensional projections in $M_2$ are parametrised by the
unit sphere $S_2$.

\smallskip\noindent\emph{Notation:}
For $a=(a_1,a_2,a_3)\in\RR^3$ let us write
   $$\s(a):=\left(\matrix{a_3&a_1-ia_2\cr a_1+ia_2&-a_3\cr}\right)
           =a_1\s_1+a_2\s_2+a_3\s_3\;,$$
where $\s_1,\s_2$ and $\s_3$ are the {\it Pauli matrices}
   $$\s_1:=\left(\matrix{0&1\cr1&0\cr}\right)\;,\quad
     \s_2:=\left(\matrix{0&-i\cr i&0\cr}\right)\;,\quad
     \s_3:=\left(\matrix{1&0\cr0&-1\cr}\right)\;.$$
We note that for all $a,b\in\RR^3$ we have
\begin{equation}
\s(a)\s(b)=\inp ab\done+i\s(a\times b)\;.
\label{sigmamult}
\end{equation}
We may now write (\ref{projunitvectorR3}) as
\begin{equation}
E(a):=\half(\one+\s(a)),\quad(\norm a=1)\;.
\label{parE}
\end{equation}

\noindent
In the same way the possible states on $M_2$ can be calculated.
We find that
\begin{equation}
\ph(A)=\tr(\r A)\where\r=\r(a):=\half(\one+\s(a)),\quad\norm a\le1\;.
\label{parrho}
\end{equation}

The probability of the event $E(a)$ in the state $\r(b)$ is given by
   $$\tr(\r(b)E(a))=\half(1+\inp ab)\;.$$
The events $E(a)$ and $E(b)$ are compatible if and only if $a=\pm b$.
Moreover we have for all $a\in S_2$:
   $$E(a)+E(-a)=\one\;,\quad E(a)E(-a)=0\;.$$

\noindent\emph{Interpretation:} The state of the qubit is given by a
vector $b$ in the three-dimensional unit ball. For every $a$ on the
unit sphere we can say with probability one that of the two events
$E(a)$ and $E(-a)$ exactly one will occur, $E(a)$ having probability
$\half(1+\inp ab)$. So we have a classical coin toss (with
probability for heads equal to $\half(1+\inp ab)$) for every
direction in $\RR^3$. The coin tosses in different directions are
incompatible. (See Fig. \ref{Bloch})

\sn The quantum coin toss is realised in nature: apart from photon
polarization (Section \ref{qubitphoton}), the spin direction of a
particle with total spin $\half$ behaves in this way.

\subsection{Photons}

\noindent There is a second natural way to parametrize the
one-dimensional projections in $M_2$, which is closer to the
description of polarization of photons, as treated in Section
\ref{qubitphoton}.


The projection onto the one-dimensional subspace spanned by the
unit vector $(\cos\a,e^{i\ph}\sin\a)$ mentioned in equation
(\ref{unitvectorC2}) of that section is given by
\begin{equation}\label{projunitvectorC2}
F(\a,\ph)=\left(\matrix{\cos^2\a&e^{-i\ph}\cos\a\sin\a\cr
                          e^{i\ph}\cos\a\sin\a&\sin^2\a\cr}\right)
\label{Falfa}
\end{equation}
Equating this projection to $E(x,y,z)$ in (\ref{projunitvectorR3})
we obtain the relations
   $$x=\sin2\a \cos\ph\;;\quad y=\sin2\a \sin\ph\;;\quad z=\cos2\a\;,$$
which define a mapping between the polarization states of a photon and
the points of the unit sphere in $\RR^3$, called the {\it Bloch sphere}
in this context.


\afbeelding{6cm}{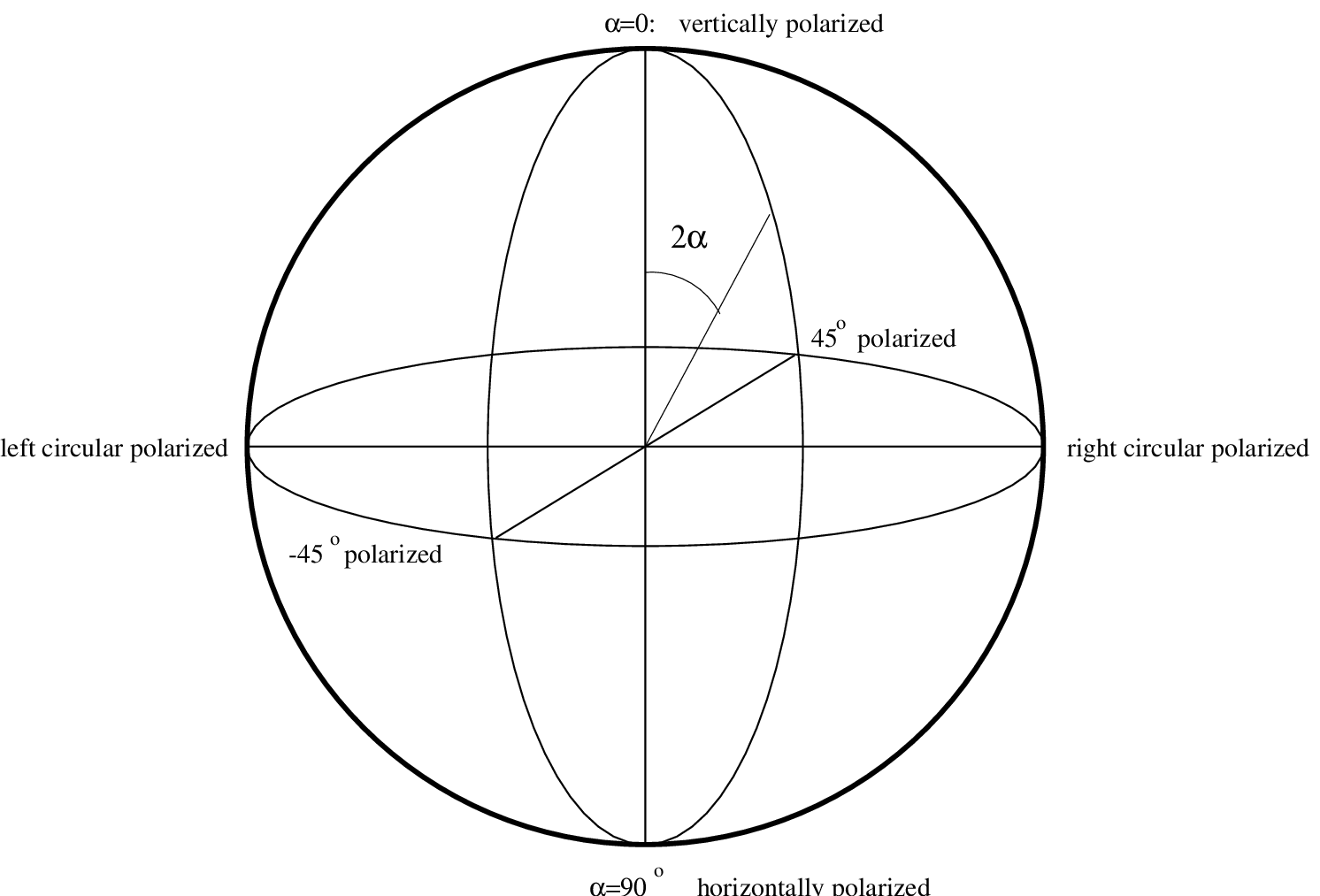}{Bloch sphere of the qubit}{Bloch}

%
%
\sn
In particular,
the projection $F(\a,0)$ onto the line in $\CC^2$ with real slope
$\tan\a$ with $\a\in[-\pi/2,\pi/2)$
is given by
\begin{equation}
F(\a,0)=\left(\matrix{\cos^2\a&\cos\a\sin\a\cr
                          \cos\a\sin\a&\sin^2\a\cr}\right)
       =E(\sin2\a,0,\cos2\a)\;.
\label{Falfa}
\end{equation}

Finally,
any atomic or molecular system,
only two energy levels of which are of importance in the experiment,
can be described by some $(M_2,\ph)$.

\sn\emph{Exercise:}
Let $f:\CC\cup\{\infty\}\to S_2$ be given by
\begin{eqnarray*}
f(0)&:=&(0,0,1)\;;\cr
         f(\infty)&:=&(0,0,-1)\;;\cr
      f(re^{i\ph})&:=&(\sin\th\cos\ph,\sin\th\sin\ph,\cos\th)\cr
      &\quad&\with\th=2\arctan r,\quad r\in(0,\infty),\ph\in[0,\pi)\;.\cr
\end{eqnarray*}
Show that $E(f(z))$ is the one-dimensional projection onto the line in
$\CC^2$ with slope $z\in\CC$.

\section{Operations on Probability Spaces}
\label{sec:Operations}

Our main objects of study will be {\it operations} on probability spaces.
This means that we shall focus attention on the input-output aspect
of probabilistic systems.

\subsection{Operations on classical probability spaces}

It could be maintained that operations are already the core of {\it classical}
probability.
We start with a definition on the level of points.

\begin{definition}
By an {\it operation} from a finite classical probability space $\O$ to a
finite classical probability space $\O'$ we mean an $\O\times\O'$
transition matrix,
i.e. a matrix $(t_{\o\o'})$ of nonnegative numbers satisfying
   $$\forall_{\o\in\O}:\quad\sum_{\o'\in\O'}t_{\o\o'}=1\;.$$
\end{definition}

\noindent\emph{Examples:}
\begin{enumerate}
\item
Let $\t$ be a bijection $\O\to\O'$.
We may think of rearranging a deck of cards, ($\O=\O'=\{\hbox{cards}\}$),
or the time evolution of a mechanical system ($\O=\O'=\hbox{phase space}$),
or the shift on sequences of letters,
or just some relabeling of the outcomes of a statistical experiment.
The associated matrix is
   $$t_{\o\o'}:=\cases{1&if $\o'=\t(\o)$,\cr 0&otherwise.\cr}$$

\item
Let $X:\O\to\O'$ be surjective.
We think of $X$ as an $\O'$-valued {\it random variable},
where $\O'$ is usually some subset of $\RR$ or $\RR^n$ or so.
The associated operation is that of `measuring $X$' or
`forgetting everything about $\o$ except the value of $X$'.
The associated matrix
is again
   $$t_{\o\o'}:=\cases{1&if $\o'=X(\o)$,\cr 0&otherwise.\cr}$$

\item
An inverse to the operation of Example 2 is given by
   $$t_{\o'\o}:=\cases{{{\pi(\{\o\})}\over{\pi(X^{-1}(\{\o'\}))}}
                        &if $\o'=X(\o)$,\cr 0&otherwise.\cr}$$
Here $\pi$ is some probability distribution, which we assume
to be everywhere nonzero.
This operation describes the {\it immersion} of a system $\O'$
into the larger system $\O$.
\end{enumerate}

\mn
It can be shown that every transition matrix can be decomposed
as a product of matrices of the types 3, 1 and 2.
So every operation can be decomposed as an immersion, followed
by a rearrangement and a restriction.
Such a decomposition is called a {\it dilation} of the operation in
question.

\subsection{Quantum Operations}
If $\A$ is a unital $*$-algebra describing a quantum system,
then we denote by $\A^*$ the dual of $\A$,
and by $\Astat$ the positive normalized functionals,
i.e. the {\it states} on $\A$.
By $M_n(\A)$ we denote the unital $*$-algebra of all $n\times n$-matrices
with entries in $\A$. Note that $M_n(\A)$ is isomorphic to $M_n\ten\A$.

\smallskip
Now suppose that we perform a physical operation which takes as input
a state on the system $\A$,
and yields as its output a state on the system $\B$.
Which maps $f:\Astat\to\Bstat$ can occur as descriptions of such an operation?
We formulate three natural requirements.

\begin{enumerate}
\item
$f$ must be an affine map.
This means that for all $\rho,\theta\in\Astat$ and all $\lambda\in[0,1]$:
   $$\l f(\rho)+(1-\l)f(\th)=f\bigl(\l\rho+(1-\l)\th\bigr)\;.$$
This requirement is a consequence of the {\it stochastic equivalence principle}
which states that a system which is in state $\rho$ with probability $\l$
and in state $\th$ with probability $1-\l$ can not be distinguished
from a system in the state $\l\rho+(1-\l)\th$.

\noindent
A map $f$ satisfying this condition can be extended to a unique
linear map $\A^*\to\B^*$,
since every element of $\A^*$ can be written as a linear combination of
(at most four) states on $\A$.
So $f$ must be the adjoint of some linear map $T:\B\to\A$.
We shall henceforth write $T^*$ instead of $f$.

\item
Of course, $f=T^*$ must still map $\Astat$ to $\Bstat$: for all $\rho\in\A^*$,
\begin{eqnarray*}
&\tr(T^*\rho)=\tr(\rho)\;;\\
&T^*\rho\ge0\quad\hbox{if}\quad\rho\ge0\;.
\end{eqnarray*}

\item
It would seem at first sight that nothing more can be said a priori
about $T^*$.
However, it was realised in the early 1980's by Karl Kraus \cite{Kraus}
that the positivity property has to be strengthened in quantum mechanics:
if the system under consideration is in a combined state
with some other system,
then after performing the operation $T^*$ on the former system,
the whole combination must still be in some (positive) state.

\noindent
Surprisingly,
this is not automatic in the quantum situation,
where `entanglement', as treated in Section \ref{sec:Bell},
can occur between the two systems. See Example \ref{PosNonCP} below.

\noindent
Therefore this stronger form of positivity must be added as a requirement:
For all $n\in\NN$:
   $$\id_n\ten T^*\quad\hbox{ maps states on }\quad M_n\ten\A
                  \quad\hbox{to states on }\quad M_n\ten\B\;.$$
\end{enumerate}

\sn
Requirement (3) is called {\it complete positivity} of the map $T^*$
(or $T$ for that matter).

\sn
Summarizing we arrive at the following definition,
which we shall formulate in the contravariant, `Heisenberg' picture.

\begin{definition}
A linear map $T:\B\to\A$ is called an {\it operation} (from $\A$ to $\B$!)
if the following conditions hold:

\begin{enumerate}
\item
$T(\one_\B)=\one_\A$;

\item
$T$ is completely positive,
i.e. $\id_n\ten T$ is positive $M_n(\B)\to M_n(\A)$ for all $n\in\NN$.
\end{enumerate}
\end{definition}

\noindent
Here $M_n(\A)$ stands for the algebra of $n\times n$ matrices with
entries in $\A$. This algebra is isomorphic to $M_n\otimes\A$.

\begin{example}\label{PosNonCP}

{\it A map which is positive, but not completely positive:}

\noindent
Let $\A:=M_2$ and let
   $$T^*:\A^*\to\A^*:\left(\matrix{a&b\cr c&d\cr}\right)\mapsto
               \left(\matrix{a&c\cr b&d\cr}\right)$$
be the transposition map.
Then $T^*$ is linear, positive, and preserves the trace.
However, $T^*$ is not completely positive since
   $$\id_2\ten T^*:{1\over2}
                \left(\matrix{1&0&0&1\cr0&0&0&0\cr0&0&0&0\cr1&0&0&1\cr}\right)
\mapsto{1\over2}\left(\matrix{1&0&0&0\cr0&0&1&0\cr0&1&0&0\cr0&0&0&1\cr}\right)\;.$$
The matrix on the left is a projection
(on the vector $(e_0\ten e_0+e_1\ten e_1)/\sqrt2\in\CC^2\ten\CC^2$;
compare the entangled state of Section 2.5.);
whereas the matrix on the left has eigenvalues $\half, \half, \half$
and $-\half$,
hence is not a valid density matrix.

\sn
However,
if $\A$ or $\B$ is abelian,
then any positive operator $T:\A\to\B$ is automatically completely positive.

\end{example}

\subsection{Examples of quantum operations}
\label{standardCP}
\begin{itemize}
\item
Let $U\in M_n$ be unitary.
Then the automorphism
$T:M_n\to M_n: A\mapsto U^*AU$ is an operation.
(See Lemma 1 below.)

\item
The $*$-homomorphism $j:M_k\to M_l\ten M_k:A\mapsto\one\ten A$
is an operation.
(See Lemma 1 below.)

\item
Let $\ph$ be a state on $M_k$.
Then the map
$E:M_l\ten M_k\to M_k: B\ten A\mapsto\ph(B)A$
is an operation.

\end{itemize}

\sn
The above examples are to be compared with those in Section 4.1.
We shall prove their validity in two Lemmas.

\begin{lemma}
If $\A\subset M_k$ and
$T:\A\to\B\subset M_l$ is a $*$-homomorphism,
i.e. if for all $A$, $B\in\A$ we have
$T(AB)=T(A)T(B)$ and $T(A^*)=T(A)^*$,
then $T$ is completely positive.
\end{lemma}

\begin{proof}
We must show that for all $n\in\NN$ the map
   $$\id_n\ten T:
     \bigl(A_{ij}\bigr)_{i,j=1}^n\mapsto\bigl(T(A_{ij})\bigr)_{i,j=1}^n$$
is positive.
Indeed, for all $\psi=(\psi_1,\cdots,\psi_n)\in(\CC^l)^n$,
putting $A=X^*X$ with $X\in M_n(\A)$:
\begin{eqnarray*}
\inp{\psi}{(\id_n\ten T)(X^*X)\psi}
     &=&\sum_{i,i'=1}^l\inp{\psi_i}{T\bigl((X^*X)_{ii'}\bigr)\psi_{i'}}\cr
     &=&\sum_{i,i'=1}^l\sum_{j=1}^n
            \inp{\psi_i}{T\bigl(X_{ji}^*X_{ji'}\bigr)\psi_{i'}}\cr
     &=&\sum_{i,i'=1}^l\sum_{j=1}^n
            \inp{\psi_i}{T(X_{ji})^*T(X_{ji'})\psi_{i'}}\cr
     &=&\sum_{j=1}^n\norms{\sum_{i=1}^l T(X_{ji})\psi_i}\ge0\;.\cr
\end{eqnarray*}
\qed
\end{proof}

\begin{lemma}
\label{simpleCP}
Let $\A\subset M_k$, $\B\subset M_l$ and let $V$ be a linear map
$\CC^l\to\CC^k$.
Then
   $$T:\A\to\B:A\mapsto V^*AV$$
is completely positive.
\end{lemma}

\begin{proof}
If $(A_{ij})_{i,j=1}^n\in M_n(\A)$ is positive,
then for all $(\psi_1,\cdots,\psi_n)\in(\CC^l)^n=\CC^n\ten\CC^l$ we have
\begin{eqnarray*}
\inp{\psi}{(\id_n\ten T)(A)\psi}
     &=&\sonij\inp{\psi_i}{T(A_{ij})\psi_j}\\
     &=&\sonij\inp{\psi_i}{V^*A_{ij}V\psi_j}\\
     &=&\sonij\inp{V\psi_i}{A_{ij}V\psi_j}\ge0\;.
\end{eqnarray*}
\qed
\end{proof}

\noindent
Lemma \ref{simpleCP} covers the third case in Example \ref{standardCP}
above since $\ph$ can be decomposed into
pure states as $\ph=\sum_i\inp\psi{\cdot\psi}$ and
   $$\ph(B)A=\sol i\l_i\inp{\psi_i}{B\psi_i}A
            =\sol i\l_i V_i^*(B\ten A)V_i\;,$$
where $V_i:\CC^k\to\CC^l\ten\CC^k:\th\mapsto \psi_i\ten\th$.

\subsection{Unraveling quantum operations}

The following important theorem, together with Proposition \ref{GenStarAlg},
characterizes all completely positive maps on finite dimensional
matrix algebras.

\begin{theorem}\label{Stinespring}
(Stinespring 1955).
Let $T$ be a linear map $M_k\to M_l$.
Then $T$ is completely positive if and only if
there exist $m\in\NN$ and operators $V_1,\ldots V_m:\CC^l\to\CC^k$
such that for all $A\in M_k$:
\begin{equation}\label{StinespringDecomp}
T(A)=\som i V^*_iAV_i\;.
\end{equation}
\end{theorem}

We shall give a proof based on a physical argument (cf. \cite{NiC}).
The system is put in an entangled state with a second system,
which for convenience we describe by the {\it opposite algebra} (see below).
Then we act on the main system with our operation $T$,
and by complete positivity we get a new state on the pair.
Surprisingly, this state fully characterizes the operation $T$.
By decomposing the state into vector states we shall obtain the unraveling
we wanted.

Let us first introduce some notation. If $\H$ is a (finite
dimensional) Hilbert space, let $\H'$ denote its {\it dual}, the
space of all linear functionals $\H\to\CC$. The elements of $\H'$
are of the form $\bar\th:\chi\mapsto\inp\th\chi$; in Dirac notation
$\bar\th$ is denoted as $\bra\th$. This dual $\H'$ is actually
isomorphic to $\H$ itself, but it is convenient to maintain the
distinction, as we shall see below. In particular, if $\H=\CC^n$,
then there is a natural action on $\H'$ of the algebra $M_n^t$, the
{\it opposite algebra} of $M_n$, which has the multiplication
reflected: $A^tB^t=(BA)^t$. The operator $A^t$ acts on $\bar\chi$ as
$A^t\bar\chi:=\bar\chi\circ A$.

\noindent
Now consider the tensor product $\H_{kl}:=\CC^k\ten(\CC^l)'$ of the
Hilbert space $\CC^k$ and the dual of $\CC^l$.
By identifying the vector $\psi\otimes\bar\th\in\H_{kl}$
with the operator $\ket\psi\bra\th$: $\chi\mapsto\inp\th\chi\cdot\psi$,
the Hilbert space $\H_{kl}$ can alternatively be viewed as the space
of all operators $\CC^l\to\CC^k$.
On this Hilbert space the
algebra $M_k\ten M_l^t$ acts naturally as follows:
   $$A\ten B^t:\psi\ten\bar\th\mapsto A\psi\ten B^t\bar{\th}
            \quad\biggl[\approx A\ket\psi\bra\th B\biggr]\;.$$
The space $\H_{ll}$ has a rotation invariant vector
(the so-called fully entangled state on $M_l\ten M_l^t$),
given by
   $$\O:={1\over\sqrt l}\sol i e_i\ten\bar{e_i}
      \quad\biggl[\approx{1\over\sqrt l}\sol i\ket{e_i}\bra{e_i}
                  =\one_l/\sqrt l\biggr]\;,$$
for {\it any} orthonormal basis $e_1,\ldots,e_l$ of $\CC^l$.
This vector has the property that
\begin{eqnarray}\label{Omegatrace}
\inp{\O}{(A\ten B^t)\O}
     &=&{1\over l}\sol i\sol j\inp{e_i\ten\bar{e_i}}
           {(A\ten B^t)e_j\ten\bar{e_j}}\cr
     &=&{1\over l}\sol i\sol j\inp{e_i}{Ae_j}\inp{\bar{e_i}}{B^t\bar{e_j}}\cr
     &=&{1\over l}\sol i\sol j\inp{e_i}{Ae_j}\inp{e_j}{B e_i}
     ={1\over l}\tr(AB)\;.
\end{eqnarray}

\sn
{\it Proof of Stinespring's Theorem.}
The `if' part follows immediately from Lemma 4.2.
For the `only if' part, assume that $T:M_k\to M_l$ is completely positive.
Let $\H_{ll}:=\CC^l\ten(\CC^l)'$ as above,
and let $\o$ denote the state
   $$\o(X):=\inp{\O}{X\O}$$
on $\B(\H_{ll})\approx M_l\ten M_l^t$.
Since $T$ is completely positive, the functional $\o_T$ on
$\B(\H_{kl})\approx M_k\otimes M_l^t$,
given by
   $$\o_T(A\otimes B^t):=\o(T(A)\otimes B^t)$$
is also a state. Decompose $\o_T$ into pure states
given by vectors $v_1,v_2,\ldots,v_m\in\H_{kl}$:
   $$\o_T(X)=\som i\inp{v_i}{Xv_i}\;.$$
Now, as noted above,
$v_i\in\H_{kl}$ can be considered as an operator $V_i:\CC^l\to\CC^k$.
We shall show that these operators satisfy the requirement
(\ref{StinespringDecomp}) of the theorem.
Indeed, for all $\psi,\th\in\CC^l$:
\begin{eqnarray*}
\som i\inp\psi{V_i^*AV_i\th}
     &=&\som i\inp{V_i\psi}{AV_i\th}\cr
     &=&\som i\inp{v_i}
         {\bigl(A\ten(\ket{\bar\psi}\bra{\bar\th})\bigr)v_i}_{\H_{kl}}\cr
     &=&\o_T\bigl(A\ten(\ket{\bar\psi}\bra{\bar\th})\bigr)\cr
     &=&\o\bigl(T(A)\ten(\ket{\bar\psi}\bra{\bar\th})\bigr)\cr
     &=&\tr\bigl(T(A)(\ket\th\bra\psi)\bigr)\cr
     &=&\inp\psi{T(A)\th}\;.\qquad\qquad\qquad\qquad\qed
\end{eqnarray*}
The second step is verified by substituting
$V_i=\sum_j\ket{\a_j^i}\bra{\b_j^i}$ with $\a_j^i\in\CC^k$,
$\b_j^i\in\CC^l$, and realizing that
$v_i=\sum_j\a_j^i\ten\bar{\b_j^i}$.

\subsection{Uniqueness of unravelings}\label{subsec:uniqueness}
(This section elaborates on a remark by Mark Fannes during the
Summer School. It can be skipped in a first reading.) The unraveling
(\ref{StinespringDecomp}) is not unique. If the matrices
$V_1,\ldots,V_m$ are linearly independent, then they are determined
by the completely positive map $T$ up to a transformation of the
form
\begin{equation}\label{RotationUnravel}
V_i':=\som j u_{ij}V_j\;,
\end{equation}
where $u$ is a unitary $m\times m$-matrix of complex numbers.
In this independent case the number $m$ of terms in the unraveling
takes its minimal value,
which we shall call the {\it rank} of the operation $T$.

\noindent In general, any number $m$ of terms, also larger than the
rank, can occur in the unraveling of $T$. But in that case the
operators $V_i$ are not linearly independent. In fact, the space
$\D$ of {\it dependencies}, given by
   $$\D:=\set{\l\in\CC^m}{\som i \bar\l_i V_i=0}\;,$$
has dimension $m-\rank(T)$,
and the matrix $u$ of (\ref{RotationUnravel}) is a partial isometry
with initial space $\D^\perp$ and final space $(\D')^\perp$,
where $\D'$ denotes the space of dependencies of the $V_i'$.

We shall now prove these statements in the context of the decomposition
of states.
From the proof of Theorem \ref{Stinespring} it is clear that they carry
over to operations.

\begin{proposition}\label{staterepr}
Let $\ph$ be a state on $\A:=M_k$, and let two
decompositions of $\ph$ into pure states be given:
\begin{equation}
   \ph(A)=\som i\inp{\psi_i}{A\psi_i}=\son j\inp{\th_j}{A\th_j}\;.
\label{decomp}
\end{equation}
Let $\D\subset\CC^m$ and $\D'\subset\CC^n$ denote the dependency spaces of
$\psi=(\psi_1,\ldots,\psi_m)$ and $\th=(\th_1,\ldots,\th_n)$ respectively.
Then $\psi$ and $\th$  are connected by a transformation of the form
   $$\th_j=\som i u_{ji}\psi_i\;.$$
where the $n\times m$ matrix $u$ describes
a partial isometry $\CC^m\to\CC^n$ with initial space $\D^\perp$
and final space $(\D')^\perp$.
In particular, if the $m$-tuple $(\psi_1,\ldots,\psi_m)$
and the $n$-tuple $(\th_1,\ldots,\th_n)$ are both sequences of independent
vectors, then $n=m$ and $u$ is unitary.
\end{proposition}

\begin{proof}
Consider $\psi$ and $\th$
as vectors in $\H:=(\CC^k)^m=\CC^m\ten\CC^k$ and
$\H':=(\CC^k)^n=\CC^n\ten\CC^k$ respectively.
Then (\ref{decomp}) can be written in the form
   $$\ph(A)=\inp{\psi}{(\one_m\ten A)\psi}=\inp{\th}{(\one_n\ten A)\th}\;.$$
Let $\L\subset\H$ and $\L'\subset\H'$ be the subspaces consisting of
the vectors $(\one_m\ten A)\psi$ and $(\one_n\ten A)\th$
respectively, where $A$ runs through the matrix algebra $\A=M_k$.
Let $U:\L\to\L'$ be given by
   $$U(\one_m\ten A)\psi:=(\one_n\ten A)\th\;.$$
Then $U$ is well-defined, isometric, and onto since
\begin{eqnarray*}
\norms{(\one_n\ten A)\th}
        &=&\inp{(\one_n\ten A)\th}{(\one_n\ten A)\th}
         =\inp{\th}{(\one_n\ten A^*A)\th}\cr
        &=&\ph(A^*A)
         =\norms{(\one_m\ten A)\psi}\;.
\end{eqnarray*}
We extend $U$ to a map $\H\to\H'$ by putting $U\chi=0$ for all
$\chi\in\H$ which are orthogonal to $\L$.

\noindent
Next, let us show that $U$ is actually of the form $u\ten\one_k$
for some partial isometry $u:\CC^m\to\CC^n$.
This is equivalent to the statement that for all $A\in M_k$:
   $$U(\one_m\ten A)=(\one_n\ten A)U\;,$$
which is true since
$(\one_m\ten A)$ leaves $\L^\perp$ invariant,
so that both sides vanish on $\L^\perp$.
And for $\chi\in\L$, i.e. for
$\chi=(\one_m\ten X)\psi$ with $X\in M_k$, we have
\begin{eqnarray*}
U(\one_m\ten A)\chi&=&U(\one_m\ten A)(\one_m\ten X)\psi
         =U(\one_m\ten AX)\psi=(\one_n\ten AX)\th\cr
        &=&(\one_n\ten A)(\one_n\ten X)\th
         =(\one_n\ten A)U(\one_m\ten X)\psi
         =(\one_n\ten A)U\chi\;.
\end{eqnarray*}
It remains to be shown that
   $$\L^\perp=\D\ten\CC^k$$
(and analogously $(\L')^\perp=\D'\ten\CC^k$).
Clearly, for all $\l\in\CC^m$ and $\mu\in\CC^k$,
\begin{equation}\label{Lperp}
   \inp{\l\ten\mu}{(\one\ten A)\psi}
   =\som i\bar{\l_i}\inp{\mu}{A\psi_i}
   =\Inp{A^*\mu}{\left(\som i\bar{\l_i}\psi_i\right)}\;.
\end{equation}
It follows that for $\l\in\D$ the vector $\l\ten\mu$ is orthogonal
to $\L$, so we have $\D\ten\CC^k\subset \L^\perp$. To prove the
converse inclusion, we first note that the orthogonal projection
onto $\L$ is $U^*U=u^*u\ten\one_k$, hence $\L=\E\ten\CC^k$ for some
subspace $\E$ of $\CC^m$. We must show that $\E^\perp\subset\D$. So
suppose that $\l\perp\E$, so that $\l\ten\mu\perp\L$ for all
$\mu\in\CC^k$. Putting $A=\one$ in (\ref{Lperp}) we find that the
left hand side, and hence the right hand side, is 0 for all $\mu$,
so $\som i\bar{\l_i}\psi_i=0$ and $\l\in\D$. \qed
\end{proof}


\subsection{Properties of quantum operations}

\noindent
When $A$ and $B$ are operators on a Hilbert space,
we mean by $A\ge B$ that the difference $A-B$ is a positive operator.
The following is an extremely useful inequality for operations.

\begin{proposition}\label{CauchySchwartz}
(Cauchy-Schwarz for operations)
Let $\A$ and $\B$ be *-algebras of operators on Hilbert spaces $\H$ and $\K$,
and let $T:\A\to\B$ be an operation.
Then we have for all $A\in\A$:
   $$T(A^*A)\ge T(A)^*T(A)\;.$$
\end{proposition}

\begin{proof}
The operator $X\in M_2\ten\A$ given by
  $$X:=\left(\matrix{A^*A&-A^*\cr-A&\one\cr}\right)
      =\left(\matrix{A&-\one\cr0&0\cr}\right)^*
       \left(\matrix{A&-\one\cr0&0\cr}\right)$$
is positive.
Since $T$ is completely positive and $T(\one)=\one$,
it follows that also
  $$(\id\ten T)(X)=\left(\matrix{T(A^*A)&-T(A)^*\cr-T(A)&\one\cr}\right)$$
is a positive operator.
Putting $\xi:=\psi\oplus T(A)\psi$ we find that
  $$\inp\xi{(\id\ten T)X\xi}=\inp\psi{\bigl(T(A^*A)-T(A)^*T(A)\bigr)\psi}$$
is positive for all $\psi\in\H$.\qed
\end{proof}

\begin{theorem}\label{MultThm}
(Multiplication Theorem)
If $T:\A\to\B$ is  an operation and $T(A^*A)=T(A)^*T(A)$ for some $A\in\A$,
then $T(A^*B)=T(A)^*T(B)$ and $T(B^*A)=T(B)^*T(A)$ for all $B\in\A$.
\end{theorem}

\begin{proof}
Take any $B\in\A$ and $\l\in\RR$. Then
   $$T\bigl((A^*+\l B^*)(A+\l B)\bigr)=
     T(A)^*T(A)+\l T(A^*B+B^*A)+\l^2 T(B^*B)\;,$$
while by Cauchy-Schwartz
\begin{eqnarray*}
T\bigl((A^*+\l B^*)&(A+\l B)\bigr)\cr
     &\ge T(A)^*T(A)+\l(T(A)^*T(B)+T(B)^*T(A))+\l^2 T(B)^*T(B))\;.
\end{eqnarray*}
This inequality holds for all $\l\in\R$ which implies
   $$T(A^*B+B^*A)\ge T(A)^*T(B)+T(B)^*T(A)\;.$$
Replacing $A$ by $iA$ and $B$ by $-iB$ shows that the opposite inequality
also holds, so we have equality.
Finally replacing only $B$ by $iB$ shows that $T(A^*B)=T(A)^*T(B)$
and $T(B^*A)=T(B)^*T(A)$.
\qed
\end{proof}

\noindent
In particular, if a Cauchy-Schwartz {\it equality} holds for an operation
$T$ then $T$ is a *-homomorphism.

\begin{theorem}\label{EmbThm}
(Embedding theorem)
Let $(\A,\ph)$ and $(\B,\psi)$ be nondegenerate quantum probabality spaces,
and let $j:\A\to\B$, $E:\B\to\A$ be operations which preserve the states.
If
   $$E\circ j=\id_\A\;,$$
then $j$ is an injective *-homomorphism and $P:=j\circ E$ is a conditional
expectation,
i.e.,
\begin{equation}\label{condexp}
   P(C_1BC_2)=C_1P(B)C_2
\end{equation}
for all $C_1,C_2\in j(\A)$ and all $B\in\B$.
\end{theorem}

\noindent
Following the language used in Section 4.1. we shall call $j$ a
{\it random variable} and $P$ the {\it conditional expectation
with respect to $\psi$, given $j$}.
Compare the following proof with that of Theorem 4.1.

\begin{proof}
For any $A\in\A$ we have by Cauchy-Schwartz
\begin{equation}\label{CSjA}
   A^*A=E\circ j(A^*A)\ge E(j(A)^*j(A))\ge E\circ j(A)^*E\circ j(A)
         =A^*A\;,
\end{equation}
so we have equalities here.
In particular
         $$\psi\bigl(j(A^*A)-j(A)^*j(A)\bigr)
    =\ph\circ E\bigl(j(A^*A)-j(A)^*j(A)\bigr)
    =0\;,$$
and as $(\B,\psi)$ is non-degenerate, $j(A^*A)=j(A)^*j(A)$,
i.e. $j$ is a *-homomorphism.
$j$ is injective since it has the left-inverse $E$.

\noindent But also from (\ref{CSjA}) we have
$E\bigl(j(A)^*j(A)\bigr)=E\circ j(A)^* E\circ j(A)$. The
Multiplication Theorem \ref{MultThm} then implies that for all
$B\in\B$ and $A_1\in\A$,
   $$E(j(A_1)^*B)=E\circ j(A_1)^*E(B)=A_1^*E(B)\;,$$
and similarly, with $A_2\in\A$:
   $$E\bigl(j(A_1)^*Bj(A_2)\bigr)=E\bigl(j(A_1)^*B\bigr)E\circ j(A_2)
                                 =A_1^*E(B)A_2\;.$$
Applying $j$ to both sides we find (\ref{condexp}).\qed
\end{proof}

\section{Quantum impossibilities}
\label{sec:QuantImp}
The result of any physical operation applied on a probabilistic
system (quantum or not) is described by a completely positive identity
preserving map from the state space of that system to the state space
of the resulting system.
This imposes strong restrictions on what can be done.
Some of these are well-known quantum principles, such as the Heisenberg
principle (`no measurement without disturbance'),
some are surprising and relatively recent discoveries (`no cloning'),
but all of them obtain quite neat formulations in the language
of quantum probability.

\subsection{`No cloning'}\label{nocloning}
In its original formulation \cite{Die,WoZ} the `No Cloning Theorem'
dealt with the reproduction of nonorthogonal vector states. Here we
give an algebraic version, which distinguishes clearly between the
classical and the quantum case.

 {\it
`Cloning'}, or
--- more mundanely
--- {\it copying} a stochastic object is an operation which takes as
input an object in some state $\rho$ and yields as its output a pair
of objects with identical state spaces, such that, if we throw away
one of them, we are left with a single object in the state $\rho$.
(Cf. Fig. \ref{copier}, which is actually not complete: the same
equality should hold with the other output line blocked.)

\begin{figure}
\centering
\includegraphics[height=2cm]{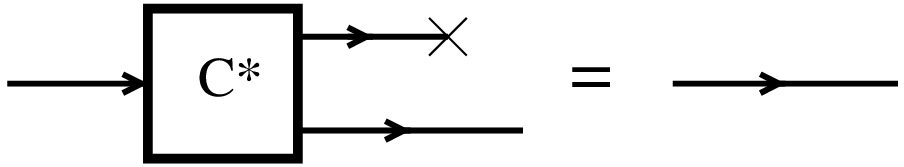}
\caption{Definition of a copier}
\label{copier}
\end{figure}

\noindent
In a formula: for all $\rho\in\Astat$:
\begin{equation}\label{cloningS}
   (\tr\ten\id)\circ C^*(\rho)=(\id\ten\tr)\circ C^*(\rho)=\rho\;.
\end{equation}
Reformulated in the Heisenberg picture:
We call an operation $C:\A\ten\A\to\A$ a {\it copying opration}
or {\it copier} if for all $A\in\A$:
\begin{equation}\label{cloning}
   C(\one\ten A)=C(A\ten\one)=A\;.
\end{equation}

\noindent
As is well known, copying presents no problem in classical physics,
or classical probability.
Here is an example of a classical copying operation.
For simplicity, let us think of the operation of copying $n$ bits.
Let $\O$ denote the space $\{0,1\}^n$ of all strings of $n$ bits,
and let $\g$ be the `copying' map $\O\to\O\times\O:\o\mapsto(\o,\o)$.
This map induces an operation
   $$C:\quad \C(\O)\times \C(\O)\to \C(\O):
       \quad Cf(\o):=f\circ\g(\o)=f(\o,\o)\;.$$
Clearly, for all $f\in\C(\O)$:
   $$C(\one\ten f)(\o)=(\one\ten f)(\o,\o)=f(\o)\;,$$
and the same holds for $C(f\ten\one)$, so (\ref{cloning}) is satisfied.
In the Schr\"odinger picture our operation looks as follows:
for any probability distribution $\pi$ on $\O$,
   $$(C^*\pi)(\nu,\o)=\d_{\nu\o}\pi(\o)\;,$$
and we see that (\ref{cloningS}) is satisfied:
   $$(\tr\ten\id)\circ C^*(\pi)(\o)
         =\sum_{\nu\in\O}\d_{\nu\o}\pi(\o)=\pi(\o)\;.$$
The following theorem says that this construction
is only possible in the abelian (i.e. commutative) case.

\begin{theorem}
(`No cloning')
Let $\A$ be a $*$-algebra of operators on a (finite dimensional) Hilbert space.
Then $\A$ admits a copying operation if and only if $\A$ is abelian.
\end{theorem}

\begin{proof}
If $\A$ is abelian, by Gel'fands Theorem (Theorem \ref{Gelfand}),
$\A$ is isomorphic to $\C(\O)$ for some finite set $\O$,
and the above construction of a copier applies.
Conversely, suppose that $C:\A\ten\A\to\A$ is a copying operation.
Then (\ref{cloning}) implies that for all $A\in\A$:
   $$C\bigl((\one\ten A)^*(\one\ten A)\bigr)=C(\one\ten A^*A)=A^*A
        =C(\one\ten A)^*C(\one\ten A)$$
Then it follows from the Multiplication Theorem \ref{MultThm} that
for all $A,B\in\A$:
\begin{eqnarray*}
AB&=&C(A\ten\one)C(\one\ten B)
               =C\bigl((A\ten\one)(\one\ten B)\bigr)\cr
              &=&C\bigl((\one\ten B)(A\ten\one)\bigr)
               =C(\one\ten B)C(A\ten\one)=BA\;.\qquad\qed
\end{eqnarray*}
\end{proof}

\subsection{`No classical coding'}\label{noclassclon}
Closely related to the above is the rule that
`quantum information cannot be classically coded':
It is not possible to operate on a quantum system,
extracting some information from it,
and then from this information reconstruct the quantum system in its original
state:
   $$\rho\in\A^*\mathop{\longmapsto}\limits^{C^*}\pi\in\B^*
                \mathop{\longmapsto}\limits^{D^*}\r\in\A^*\;.$$
We formulate this theorem in the contravariant (`Heisenberg') picture:

\begin{theorem}
Let $\A$ and $\B$ be *-algebras,
and let $C:\B\to\A$ and $D:\A\to\B$ be operations,
(`Coding' and `Decoding'), such that $C\circ D=\id_\A$.
Then if $\B$ is abelian, so is $\A$.
\end{theorem}

\begin{proof}
We have for all $A\in\A$:
\begin{eqnarray*}
A^*A&=&C\circ D(A^*A)\ge C\bigl(D(A)^*D(A)\bigr)\ge A^*A\cr
                    \hbox{and}\quad AA^*
     &=&C\circ D(AA^*)\ge C\bigl(D(A)D(A)^*\bigr)\ge AA^*\;,
\end{eqnarray*}
so that we again have equality everywhere.
If $\B$ is abelian, we have $D(A)^*D(A)=D(A)D(A)^*$,
so that $A^*A=AA^*$.
\qed
\end{proof}

\smallskip\noindent\emph{Exercise:}
Prove that, if $A^*A=AA^*$ for all $A\in\A$, then $\A$ is abelian.

\subsection{The Heisenberg Principle}
The {\it Heisenberg principle}
states --- roughly speaking --- that no information on a quantum
system can be obtained without changing its state.

\sn
In this form, the statement is not so interesting:
if we realise that the {\it state} of the system expresses the expectations
of its observables,
given the information we have on it,
it is no wonder that this state changes once we gain information!

\noindent
A more precise formulation is the following:

{\smallskip\noindent \narrower{\sl If we extract information from a
system whose algebra $\A$ is a factor (i.e. $\A\cap\A'=\CC\one$),
and if we {\it throw away} (disregard) this information, then still
it can not be avoided that some initial states are altered.}\par}

\sn
Let us work towards a mathematical formulation.

\noindent
A {\it measurement} is an operation performed on a physical system
which results in the extraction of information from that system,
while possibly changing its state.

\noindent
So a measurement is an operation
   $$M^*:\A^*\to\A^*\ten\B^*\;,$$
where $\A$ describes the physical system,
and $\B$ the output part of a measurement apparatus
which we couple to it.
$\A^*$ consists of states and $\B^*$ of probability distributions on the
outcomes.
So $\B$ will be commutative, but we do not need this property here.

\noindent
Now suppose that no initial state is altered by the measurement:
   $$(\id\ten\tr)M^*(\rho)=\rho\qquad\forall_{\rho\in\A^*}\;.$$
Suppose also that $\A$ is a factor.
We claim that no information can be obtained on $\rho$:
\begin{equation}\label{NoInform}
  (\tr\ten\id)M^*(\rho)=\th\;,
\end{equation}
where $\th$ does not depend on $\rho$.

\sn The diagram of Fig. \ref{HeisenbergPrinciple} symbolically
expresses this fact.
\begin{figure}
\centering
\includegraphics[height=5cm]{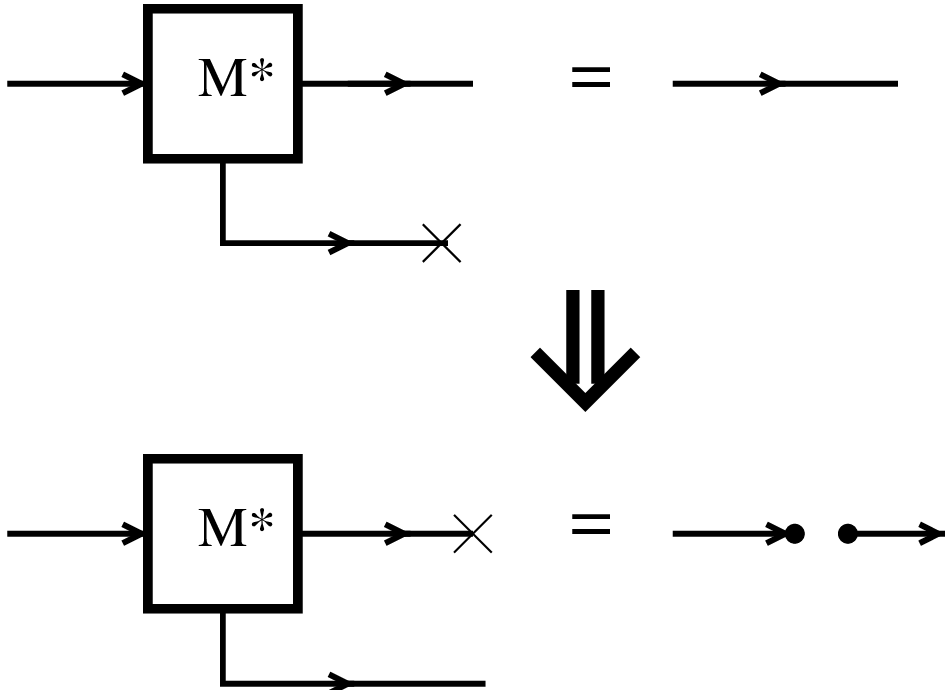}
\caption{The Heisenberg Principle}
\label{HeisenbergPrinciple}       
\end{figure}

\noindent
We again formulate and prove the theorem in the contravariant picture:

\begin{theorem}       
(Heisenberg's Principle)
Let $M$ be an operation $\A\ten\B\to\A$
such that for all $A\in\A$,
   $$M(A\ten\one)=A\;,$$
then
   $$M(\one\ten B)\in\A\cap\A'\;.$$
In particular,
if $\A$ is a factor, then for some fixed state $\th$ on $\B$:
\begin{equation}\label{NoInformHeis}
M(\one\ten B)=\th(B)\cdot\one_\A.
\end{equation}
\end{theorem}

\noindent
We note that (\ref{NoInformHeis}) implies (\ref{NoInform}),
since for all $\rho$ on $\A$ and all $B\in\B$:
   $$\bigl((\tr\ten\id)M^*\rho\bigr)(B)=\rho\bigl(M(\one\ten B)\bigr)
                          =\rho\bigl(\th(B)\one_\A\bigr)
                          =\th(B)\;.$$

\begin{proof}
As in the proof of the `no cloning' theorem we have by the multiplication
theorem for all $A\in\A$, $B\in\B$:
   $$M(\one\ten B)\cdot A=M(\one\ten B)M(A\ten\one)=M(A\ten B)\;.$$
But also,
   $$A\cdot M(\one\ten B)=M(A\ten\one)M(\one\ten B)=M(A\ten B)\;.$$
So $M(\one\ten B)$ lies in the center of $\A$.
If $\A$ is a factor, then $B\mapsto M(\one\ten B)$ is an operation
from $B$ to $\CC\cdot\one_\A$, i.e. a state on $B$ times $\one_\A$.
\qed
\end{proof}

\subsection{Random variables and von Neumann measurements}
Following the suggestion made in Section 4.2. (in particular case 2),
we define a {\it random variable} to be a *-homomorphism from
one algebra $\B$ to a (larger) algebra $\A$:
   $$\Heis j from \B to \A \;.$$
In the covariant (`Schr\"odinger') picture this describes the
operation $j^*$ of {\it restriction to} the subsystem $\B$:
   $$\Schr j from \A to \B \;.$$
An important case is when $\B=\C(\O)$ for some finite set $\O$:
then $j$ is to be viewed as an {\it $\O$-valued random variable}.
Let $\O=\{x_1,\ldots,x_n\}$.
Then $j(1_{\{x_i\}})$ is a projection, $P_i$ say, in $\A$,
with the properties that
   $$\son i P_i=\son i j(1_{\{x_i\}})=j(\one_\B)=\one_\A$$
and for $i\ne j$,
   $$P_iP_k=j(1_{\{x_i\}})j(1_{\{x_k\}})=j(1_{\{x_i\}}\cdot 1_{\{x_k\}})=0\;.$$
We interpret $P_i$ as the event `the random variable described by $j$
takes the value $x_i$'.
Note that $j$ can be written as
   $$j(f)=j\left(\son i f(x_i)1_{\{x_i\}}\right)=\son i f(x_i)P_i\;.$$
In particular, if $\O\subset\RR$, then $j$ defines a hermitian operator
   $$j(\id)=\son i x_i P_i=:X\;,$$
which completely determines $j$.

\begin{proposition}  
Let $\A$ be a finite-dimensional *-algebra with unit.
Then there is a one-to-one correnspondence between
injective *-homomorphisms
$j:\C(\O)\to\A$ for some finite $\O\subset\RR$ and self-adjoint
operators $X\in\A$, given by
   $$j(\id)=X\;.$$
\end{proposition}

\begin{proof}
If $j$ is a *-homomorphism $\C(\{x_1,\ldots,x_n\})\to\A$
with $x_1,\ldots,x_n$ real, then
   $$X:=j(\id)=\son i x_i j(1_{\{x_i\}})=:\son i x_iP_i$$
is a hermitian element of $\A$.
Conversely, if $X\in\A$ is hermitian, then let $x_1,\ldots,x_n$
be its eigenvalues.
Let $p:\CC\to\CC$ denote the polynomial
   $$p(x):=(x-x_1)\cdots(x-x_n)\;.$$
and let, for $i=1,\ldots,n$, the (Lagrange interpolation)
polynomial $p_i$ be given by
   $$p_i(x):={{p(x)}\over{(x-x_i)p(x_i)}}\;.$$
Then $p_i(x_k)=\d_{ik}p_k$, so we have
on the spectrum $\{x_1,\ldots,x_n\}$ of $X$:
   $$\son i p_i=1\And
      p_i\cdot p_k=\d_{ik}p_k\;.$$
It follows that the projections $P_i:=p_i(X)$, with $i=1,\ldots,n$, lie
in the algebra $\A$ and satisfy
   $$\son i P_i=\one\And P_iP_k=\d_{ik}P_k\;.$$
Hence, if we define
   $$j(f):=\son i f(x_i)P_i\;,$$
then $j$ is a *-homomorphism with the property that $j(\id)=X$.
Clearly, different $X$'s correspond to different $j$'s.
\qed
\end{proof}

\subsection{The joint measurement apparatus}
Let $X$ and $Y$ be self-adjoint elements of the *-algebra $\A$.
We consider $X$ and $Y$ as random variables taking values in the
spectra $\sp(X)$ and $\sp(Y)$.

\noindent
By a {\it joint measurement} $M^*$ of these random variables we mean an
operation that takes a state $\r$ on $\A$ as input,
and yields a probability distribution $\pi$ on $\sp(X)\times\sp(Y)$
as output,
in such a way that for all functions $f$ on $\sp(X)$, $g$ on $\sp(Y)$:
\begin{eqnarray*}
\r(f(X))&=&\sum_{x\in\sp(X)}\sum_{y\in\sp(Y)}\pi(x,y)f(x)\;;\cr
           \r(g(Y))&=&\sum_{x\in\sp(X)}\sum_{y\in\sp(Y)}\pi(x,y)g(y)\;.
\end{eqnarray*}
A contravariant formulation of these requirements is
\begin{eqnarray*}
               M(f\ten\one)&=&f(X)\;;\cr
              M(\one\ten g)&=&g(Y)\;.
\end{eqnarray*}

\begin{theorem}
If two random variables $X$ and $Y$ allow
a joint measurement operation, then they commute.
\end{theorem}

\begin{proof}
Let us denote by $x$ the identity function on $\sp(X)$,
and by $y$ that on $\sp(Y)$.
We apply the multiplication theorem on the measurement operation $M$,
which is supposed to exist.
Since
   $$M\bigl((x\ten\one)^*(x\ten\one)\bigr)
       =M(x^2\ten\one)=X^2=M(x\ten\one)^*M(x\ten\one)\;,$$
we have
   $$M\bigl((x\ten\one)^*(\one\ten y)\bigr)
       =M(x\ten\one)^*M(\one\ten y)=XY$$
and
   $$M\bigl((\one\ten y)^*(x\ten\one)\bigr)
       =M(\one\ten y)^*M(x\ten\one)=YX\;.$$
As $(x\ten\one)^*(\one\ten y)=x\ten y=(\one\ten y)^*(x\ten\one)$,
we have $XY=YX$.
\qed
\end{proof}

\section{Quantum novelties}
\label{sec:QuantNov}
In the previous section we saw certain strange limitations that
quantum operations are subject to.
Let us now look at the other side of the coin: some surprising possibilities.

\noindent
We leave treatment of the really sensational features
to other contributions in this volume,
such as very fast computation and secure cryptography.
Here we shall treat `teleportation' of quantum states
and `superdense coding'.

\subsection{Teleportation of quantum states}
Suppose that Alice wishes to send to Bob the quantum state $\rho$ of a qubit
over a (classical) telephone line.

\noindent
In Section \ref{noclassclon}\ (`No classical coding') we have seen that,
without any further tools, this is impossible.
If Alice were to perform measurements on the qubit, and tell the
results to Bob over the telephone, these would not enable Bob to
reconstruct the state $\rho$.

\noindent
However, suppose that Alice and Bob have been together in
the past, and that at that time they have created an entangled pair of qubits,
as introduced in Section \ref{decisive},
each taking one qubit with them.
It was discovered in 1993 by Bennett, Wootters, Peres and others,
that by making use of this shared entanglement,
Alice is indeed able to transfer her qubit to Bob.
Of course, she cannot avoid destroying the original state $\rho$
in the process;
otherwise Alice and Bob would have copied the state
$\rho$, which is impossible by Theorem \ref{nocloning} (`no cloning').
It is for this reason that the procedure is called `teleportation'.

\sn
We illustrate the procedure in a picture.


\begin{figure}
\centering
\includegraphics[height=3cm]{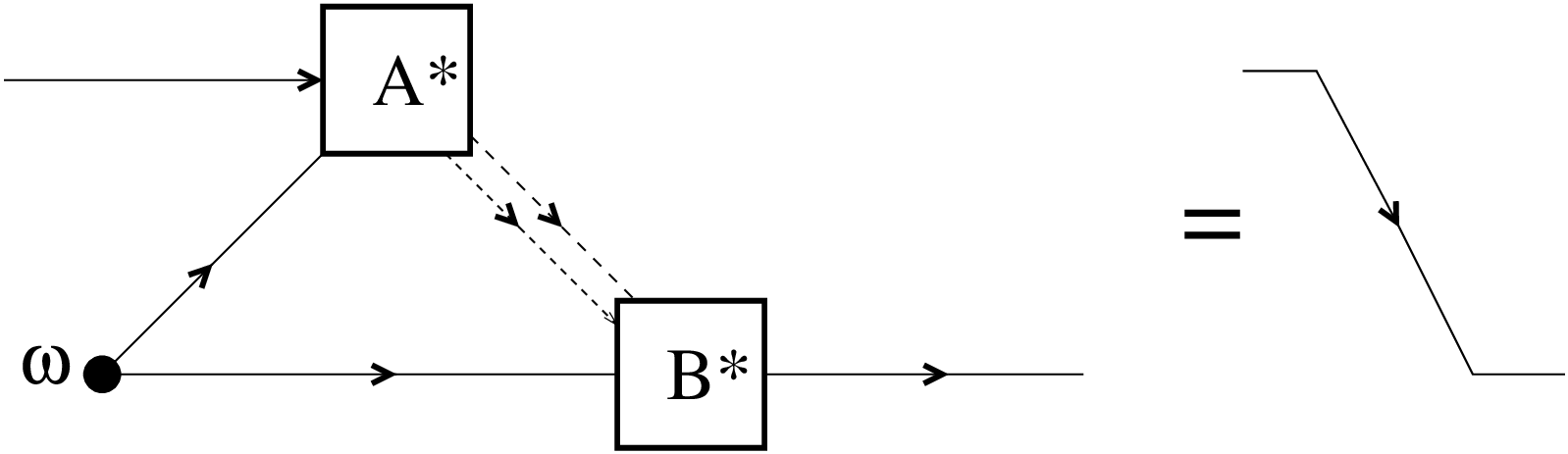}
\caption{Teleportation based on shared entanglement}
\label{TeleportS}
\end{figure}

\noindent
Here $\o$ is the fully entangled state $X\mapsto\inp\O{X\O}$ on $M_2\ten M_2$
(see the proof of Theorem \ref{Stinespring}, Stinespring's theorem.)

\sn
The procedure runs as follows.
Alice possesses two qubits, one from the entangled pair,
and one which she wishes to send to Bob.
She performs a von Neumann measurement on these two qubits along the
four {\it Bell projections}
\begin{eqnarray*}
Q_{00}&:=&
{1\over2}\left(\matrix{1&0&0&1\cr0&0&0&0\cr0&0&0&0\cr1&0&0&1\cr}\right)
\;;\qquad
Q_{01}:=
{1\over2}\left(\matrix{1&0&0&-1\cr0&0&0&0\cr0&0&0&0\cr-1&0&0&1\cr}\right)
\;;\cr
Q_{10}&:=&
{1\over2}\left(\matrix{0&0&0&0\cr0&1&1&0\cr0&1&1&0\cr0&0&0&0\cr}\right)
\;;\qquad
Q_{11}:=
{1\over2}\left(\matrix{0&0&0&0\cr0&1&-1&0\cr0&-1&1&0\cr0&0&0&0\cr}\right)\;.
\end{eqnarray*}
The operation performed by Alice has the contravariant description:
   $$A:\C_2\ten\C_2\to M_2\ten M_2:\quad A(e_i\ten e_j):=Q_{ij}\;,$$
The two bits Alice obtains in this way
--- $(i,j)$ say ---
she sends to Bob over the telephone.
He then takes his own qubit from the entangled pair,
and if $j=1$ performs the `phase flip' operation
   $$Z:\left(\matrix{\rho_{00}&\rho_{01}\cr\rho_{10}&\rho_{11}\cr}\right)
\mapsto\left(\matrix{\rho_{00}&-\rho_{01}\cr-\rho_{10}&\rho_{11}\cr}\right)
=\left(\matrix{1&0\cr0&-1\cr}\right)
\left(\matrix{\rho_{00}&\rho_{01}\cr\rho_{10}&\rho_{11}\cr}\right)
\left(\matrix{1&0\cr0&-1\cr}\right)\;,$$
and if $j=0$ he does nothing.
Then, if $i=1$ he performs the `quantum {\caps not}' operation
   $$X:\left(\matrix{\rho_{00}&\rho_{01}\cr\rho_{10}&\rho_{11}\cr}\right)
\mapsto\left(\matrix{\rho_{11}&\rho_{10}\cr\rho_{01}&\rho_{00}\cr}\right)
=\left(\matrix{0&1\cr1&0\cr}\right)
\left(\matrix{\rho_{00}&\rho_{01}\cr\rho_{10}&\rho_{11}\cr}\right)
\left(\matrix{0&1\cr1&0\cr}\right)\;,$$
and if $i=0$ he does nothing.
In the Heisenberg picture,
the result of Bob's actions is the operation
   $$B:M_2\to\C_2\ten\C_2\ten M_2:\quad
       M\mapsto M\oplus\s_3M\s_3\oplus\s_1M\s_1\oplus\s_2M\s_2\;,$$
where $\s_1:=\left(\matrix{0&1\cr1&0\cr}\right)$,
$\s_2:=\left(\matrix{0&-i\cr i&0\cr}\right)$, and
$\s_3:=\left(\matrix{1&0\cr0&-1\cr}\right)$,
are Pauli's spin matrices.
Bob ends up with a qubit in exactly the same state as Alice wanted
to send.

\noindent
We formulate this result in the Heisenberg picture.

\begin{proposition}\label{PropTeleport}
The state $\o$ and the operations $A$ and $B$ described above satisfy
$$\displaystyle (\id_{M_2}\ten\o)\circ(A\ten\id_{M_2})\circ B=\id_{M_2}\;.$$
\end{proposition}

\begin{proof}
We just calculate for $M\in M_2$:
\begin{eqnarray*}
             M &\mapsby{B}& M\oplus\s_3M\s_3\oplus\s_1M\s_1\oplus\s_2M\s_2\cr
               &\mapsby{A\ten\id}&
  (Q_{00}\ten M)+(Q_{01}\ten\s_3M\s_3)+(Q_{10}\ten\s_1M\s_1)
 +(Q_{11}\ten\s_2M\s_2)\cr
  &=&{1\over2}\left(\matrix{M+\s_3M\s_3&0&0&M-\s_3M\s_3\cr
                           0&\s_1M\s_1+\s_2M\s_2&\s_1M\s_1-\s_2M\s_2&0\cr
                           0&\s_1M\s_1-\s_2M\s_2&\s_1M\s_1+\s_2M\s_2&0\cr
                           M-\s_3M\s_3&0&0&M+\s_3M\s_3\cr}\right)\cr
  &=&\left(\matrix{m_{00}&0&0&0&|&0&0&0&m_{01}\cr
                  0&m_{11}&0&0&|&0&0&m_{10}&0\cr
                  0&0&m_{11}&0&|&0&m_{01}&0&0\cr
                  0&0&0&m_{00}&|&m_{01}&0&0&0\cr
                  -&-&-&-&|&-&-&-&-\cr
                  0&0&0&m_{10}&|&m_{11}&0&0&0\cr
                  0&0&m_{01}&0&|&0&m_{00}&0&0\cr
                  0&m_{01}&0&0&|&0&0&m_{00}&0\cr
                  m_{10}&0&0&0&|&0&0&0&m_{11}\cr}\right)\cr
  &\mapsby{\id\ten\o}&\left(\matrix{m_{00}&m_{01}\cr m_{10}&m_{11}\cr}\right)
  =M\;.\qquad\qquad\qquad\qquad\qquad\qquad\qquad\qquad\qquad\qed
\end{eqnarray*}

\end{proof}

\goodbreak\noindent
Teleportation has been carried out succesfully in the lab
by Zeilinger et al. in Vienna in 1997 using polarized photons,
and by other experimenters using different techniques later.

\noindent
For the sake of such experiments explicit operations have been developed that
form the `building blocks' of the diversity of quantum operations
needed.
For example the operation performed by Alice to prepare the teleportation
of a qubit can be decomposed into an interaction and a measurement.
Let $j$ be the ordinary measurement operation of a qubit:
 $$j:\C_2\to M_2:
     \qquad (f_0,f_1)\mapsto\left(\matrix{f_0&0\cr0&f_1}\right)\;.$$
Let $H$ denote the {\it Hadamard gate},
which acts on states or observables by multiplication on the left and on the
right by the {\it Hadamard matrix}
${1\over\sqrt2}\left(\matrix{1&1\cr1&-1}\right)$,
and let $C$ denote the {\it controlled \caps not} gate
$M_2\ten M_2\to M_2\ten M_2$ which sandwiches a matrix with
   $$\left(\matrix{1&0&0&0\cr
                   0&0&0&1\cr
                   0&0&1&0\cr
                   0&1&0&0\cr}\right)\;.$$
The operation $C$ performs a {\caps not} operation on the first qubit
provided that the second is a 1.
In diagrams:

\begin{figure}
\centering
\includegraphics[height=2.5cm]{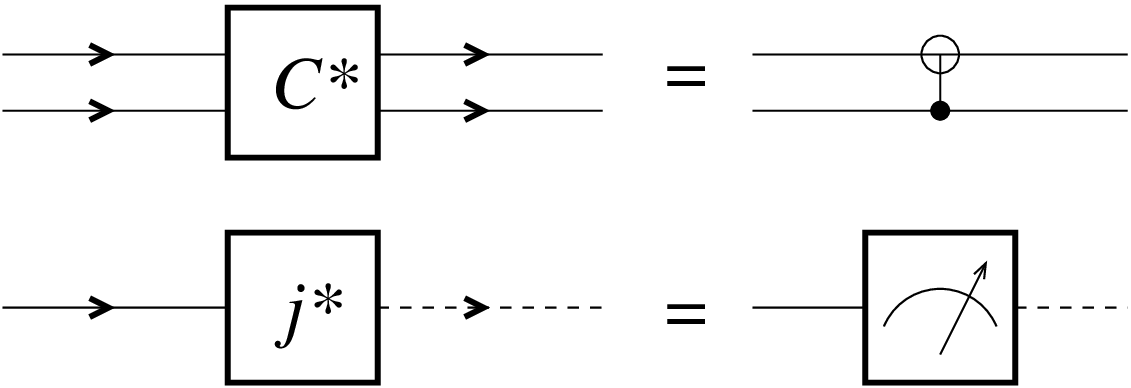}
\caption{Conventional signs used for the $C$ and $j$ operations}
\end{figure}

\noindent
Check that, using the above building blocks,
the procedure of quantum teleportation can be charted as in Figure
\ref{FigTeleport}.

\begin{figure}
\centering
\includegraphics[height=5cm]{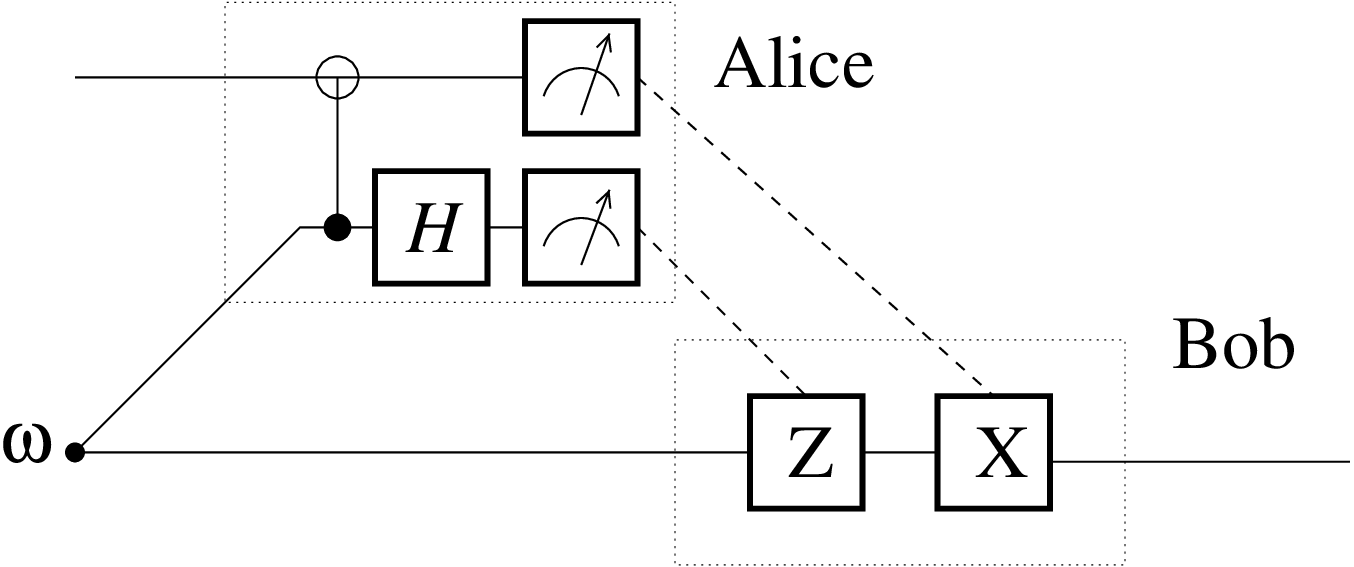}
\caption{More detailed scheme of teleportation}
\label{FigTeleport}
\end{figure}

\subsection{Superdense coding}
We have seen that Alice can `teleport' a qubit using two classical bits,
given a pre-entangled qubit pair.
A kind of converse is also possible:
Bob can communicate two classical bits to Alice by sending her a single
qubit,
again given a shared pre-entangled qubit pair.
(We have interchanged the roles of Alice and Bob here because it turns out
that in that case they can continue using exactly the same equipment
as they used for teleportation!)

\begin{figure}
\centering
\includegraphics[height=3cm]{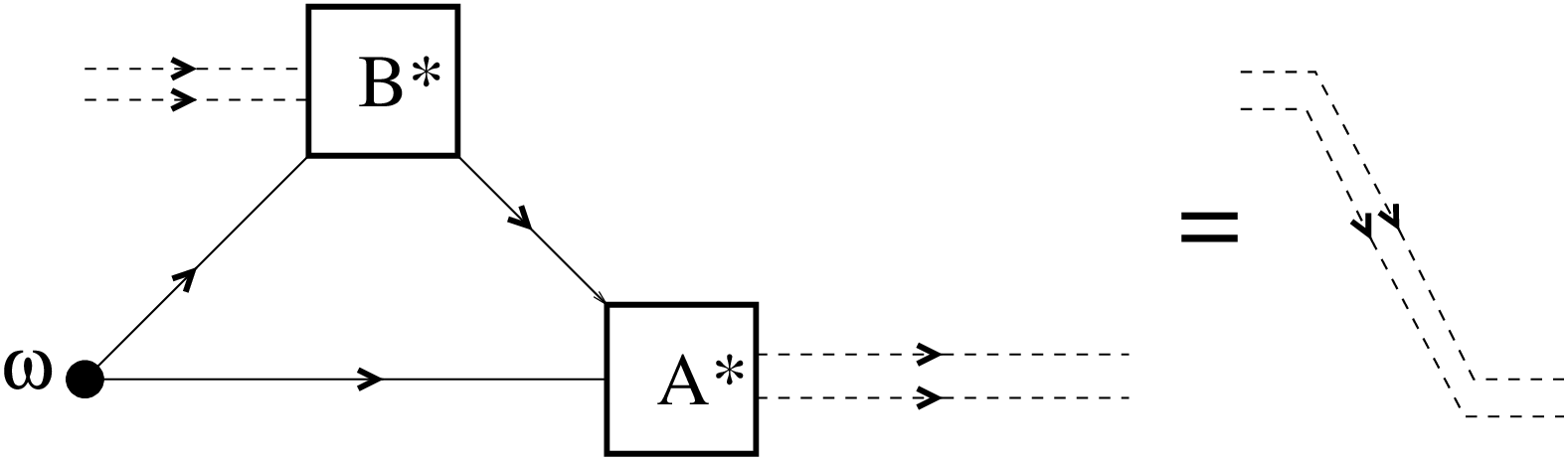}
\caption{Superdense coding: two bits in a single photon}
\label{FigSuperdense}
\end{figure}

\begin{proposition}\label{Superdense}
Taking $\o$, $A$ and $B$ as in Proposition \ref{PropTeleport},
we have
   $$(\id_{\C_2\ten\C_2}\ten\o)\circ(B\ten\id_{M_2})\circ A
            =\id_{\C_2\ten\C_2}\;.$$
\end{proposition}

\noindent
We leave the proof as an exercise.


%
%

%

\begin{thebibliography}{99.}



\bibitem{Aspect}
Alain Aspect, Jean Dalibard, G\'erard Roger: Experimental test of
Bell's inequalities using time-varying analysers. Phys. Rev. Lett.
\textbf{49} (1982), 1804--1807.

\bibitem{Bell}
John S. Bell: On the Einstein-Podolsky-Rosen paradox. Physics,
\textbf{1}, 195--200, 1964.

\bibitem{Bohm}
David Bohm: A suggested interpretation of quantum theory in terms of
`hidden variables'. Phys. Rev. \textbf{85} (1952), 180--193.

\bibitem{Bruss}
Dagmar Bruss: Quantum cryptography.
This Volume.

\bibitem{Die}
Dennis Diecks: Communication by EPR devices. Phys. Lett. A
\textbf{92} (1982), 271--272.

\bibitem{Kempe}
Julia Kempe: Quantum computation.
This Volume.

\bibitem{Kraus}
Karl Kraus: States, Effects and Operations. Lecture Notes in
Physics, Vol. \textbf{190} Springer-Verlag, Berlin 1983.

\bibitem{KuMa}
Burkhard~K\"ummerer, Hans~Maassen: Elements of Quantum Probability.
In: Quantum Probability Communications X, pp. 73--100.

\bibitem{Nelson}
Edward Nelson:
Dynamical theories of Brownian motion.
Princeton University Press,
Princeton, New Jersey, 1967.


\bibitem{NiC} Michael A. Nielsen, Isaac L. Chuang:
\textit{Quantum Computation and Quantum Information}, Cambridge
University Press, 2000.

\bibitem{Petz} Denez Petz:
Algebraic Methods for Quantum Mechanics, This Volume.


\bibitem{WoZ}
William K. Wootters, Wojciech H. Zurek: A single quantum cannot be
cloned. Nature \textbf{299} (1982), 802--803.

\end{thebibliography}
%

\end{document}